\documentclass[aps,prx,twocolumn,superscriptaddress,longbibliography]{revtex4-2}

\usepackage{silence}
\usepackage{fontenc}
\usepackage{amsmath}
\usepackage{amsfonts,amssymb,amsthm, bbm,braket,mathtools}
\usepackage{mathtools}
\usepackage{xcolor}
\usepackage{graphicx}   % need for figures
\usepackage{subfigure}  % and subfigures
\usepackage{amsbsy} % for bold greek letters
\usepackage[bold]{hhtensor}
\usepackage{times}
\usepackage{comment}
\usepackage{algorithm}
\usepackage{algpseudocode}
\usepackage{tikz}
\makeatletter
\renewcommand{\ALG@name}{Protocol}
\makeatother

\DeclareMathOperator*{\avg}{\text{\LARGE \raisebox{-0.2em}{$\mathbb{E}$}}}
\DeclareMathOperator*{\smallavg}{\text{{$\mathbb{E}$}}}

\newcounter{protocol}

\theoremstyle{remark} % theorem style to comply with NatComm guidelines

\usepackage{multirow}
\newtheorem{theorem}{Theorem}
\newcommand{\myleft}{\mathopen{}\mathclose\bgroup\left}
\newcommand{\myright}{\aftergroup\egroup\right}
\newcommand{\e}{\ensuremath\mathrm{e}}

\DeclareMathOperator{\Tr}{Tr}

\DeclareMathOperator{\SU}{SU}%special unitary group SU(n)

\DeclareMathOperator\tr{Tr}

\newcommand{\CC}{\mathbb{C}}

\newcommand{\1}{\mathds{1}}
\newcommand{\EE}{\mathbb{E}}

\newcommand{\md}[1]{\mathbb{#1}}

\newcommand{\mc}[1]{\mathcal{#1}}

\newcommand{\ct}{{}^\dagger}

\newcommand{\ad}{^\dagger}%symbol for conjugate transpose
\newcommand{\norm}[1]{\left\Vert #1 \right\Vert} %norm with variable height
\newcommand{\ketbra}[2]{\ket{#1} \!\! \bra{#2}}

\newcommand{\sandwich}[3]
  {\left\langle  #1 \right| #2 \left| #3 \right\rangle}

\newcommand{\kett}[1]{|{#1}{\rangle\!\rangle}}
\newcommand{\braa}[1]{{\langle\!\langle}{#1}|}

\newcommand{\tn}[1]{^{\otimes#1}} % Command for tensor power notation
\newcommand{\dens}[1]{\ket{#1}\!\!\bra{#1}}

\newcommand{\gr}{\ensuremath{\md{G}}}

\newcounter{example}[section]

\newcommand{\braakett}[2]{\langle\!\langle#1|#2\rangle\!\rangle}

\usepackage{aliascnt}
\usepackage[breaklinks=true]{hyperref}
\usepackage{cleveref}

\newaliascnt{lemma}{theorem}
\newtheorem{lemma}[lemma]{Lemma}
\aliascntresetthe{lemma}
\crefname{lemma}{lemma}{lemmas}
\hypersetup{
  colorlinks   = true, %Colours links instead of ugly boxes
  urlcolor     = blue, %Colour for external hyperlinks
  linkcolor    = blue, %Colour of internal links
  citecolor   = red %Colour of citations
}

\newcommand{\fu}{Dahlem Center for Complex Quantum Systems, Freie Universit{\"a}t Berlin, 14195 Berlin, Germany}

\newcommand{\tii}{Quantum Research Center, Technology Innovation Institute (TII), Abu Dhabi}

\definecolor{ingo}{rgb}{.1,.5,.1}
\newcommand{\ir}[1]{{\color{black}#1}}
\newcommand{\je}[1]{{\color{cyan}#1}}

\newcommand{\jk}[1]{{\color{black}#1}}

\newcommand{\newtext}[1]{\textcolor{black}{#1}}

\makeatletter % metadata for hyperref
 \hypersetup{pdftitle = {Shadow estimation of gate-set properties from random sequences},
       pdfauthor = {Jonas Helsen, Marios Ioannou, Ingo Roth, Jonas Kitzinger,  Emilio Onorati, Albert H. Werner, Jens Eisert},
       pdfsubject = {Quantum physics},
       pdfkeywords = {quantum, state, process, gate, set, gate-set, tomography, characterization, sequences, sequence,
              certification, verification, validation,
              direct, shadow, fidelity, estimation,
              interleaved, randomized, benchmarking,
              cross-entropy, XEB,
              concentration, inequality, inequalities, tail, bound, bounds,
              variance, estimator, median, means, unbiased,
              complex, projective, spherical, unitary, design, designs, t-design, t-designs, k-design, k-designs, random,
              Clifford, group,
            Schur-Weyl duality, Schur's lemma,
              Choi-Jamiolkowski, isomorphism,
              trace, norm, distance, average gate, entanglement, unitarity,
              machine, learning, optimization, correlator, uniform}
      }
\makeatother

\begin{document}

%---------------------------------------------------------------------------------------------------------------------
\title{Shadow estimation of gate-set properties from random sequences}

\author{J.~Helsen$^\ast$\!\!}
\affiliation{QuSoft, Centrum Wiskunde \& Informatica (CWI), Amsterdam, The Netherlands}
\email{corresponding authors: jonas1helsen@gmail.com, ingo.roth@tii.ae, jense@zedat.fu-berlin.de}
\affiliation{Korteweg-de Vries Institute for Mathematics, University of Amsterdam, Amsterdam, The Netherlands}
\author{M.~Ioannou}
\affiliation{\fu}
\author{J.~Kitzinger}
\affiliation{\fu}
\affiliation{Humboldt-Universit\"at zu Berlin, Institut f\"ur Physik,  12489 Berlin, Germany}
\author{E.~Onorati}
\affiliation{\fu}
\affiliation{University College London, Department of Computer Science, London, United Kingdom}
\affiliation{Technische Universit\"at M\"unchen, Fakult\"at f\"ur Mathematik, M\"unchen, Germany}
\author{A.~H.~Werner}
\affiliation{{Department of Mathematical Sciences, University of Copenhagen, 2100 K{\o}benhavn, Denmark}}
\affiliation{{NBIA, Niels Bohr Institute, University of Copenhagen, Blegdamsvej 17, 2100 K{\o}benhavn, Denmark}}
\author{J.~Eisert}
\affiliation{\fu}
\affiliation{Helmholtz-Zentrum Berlin f{\"u}r Materialien und Energie, 14109 Berlin, Germany}
\affiliation{Fraunhofer Heinrich Hertz Institute, 10587 Berlin, Germany}
\author{I.~Roth$^\ast$}
\affiliation{\tii}
%\affiliation{\fu}

\maketitle

{\bf
With quantum computing devices increasing in scale and complexity, there is a growing need for tools that obtain precise diagnostic information about quantum operations. However, current quantum devices are only capable of short unstructured gate sequences followed by native measurements. We accept this limitation and turn it into a new paradigm for characterizing quantum gate-sets. A single experiment---random sequence estimation---solves a wealth of estimation problems, with all complexity moved to classical post-processing. We derive robust channel variants of shadow estimation with close-to-optimal performance guarantees and use these as a primitive for partial, compressive and full process tomography as well as the learning of Pauli noise. We discuss applications to the quantum gate engineering cycle, and propose novel methods for the optimization of quantum gates and diagnosing cross-talk.}

\section{Introduction}

Recent years have seen the rapid development of quantum computing devices to unprecedented
system sizes. These devices are still noisy and of limited computational power, but
go substantially beyond what was conceivable not very long ago. In order to scale even further to larger and more accurate devices, it is key to develop tools for efficiently characterizing quantum operations \cite{Roads,Superconducting} at scale.
Besides providing crucial actionable advice for the practitioner, the characterization of quantum operations
is also important for developing an in-depth theoretical understanding of the actual capabilities of quantum devices
 and for providing a fair comparison between different types of devices, and with classical computing power on the same tasks \je{\cite{GoogleSupremacy,BarakChouGao:2020,SupremacyReview}.}
Over the years, many protocols for characterizing quantum operations have been developed \cite{BenchmarkingReview,PRXQuantum.2.010201,Randomized}.

That said, while a wealth of theoretical ideas for benchmarking, verification and tomographic recovery have been suggested,
only a few of them are relevant in practice. With present quantum devices, only
relatively short gate sequences can be implemented on qubit arrays, followed by a native measurement at the end
of the circuit that typically suffers from sizeable read-out noise.
With these limitations the most prominent protocols for characterizing digital quantum gates fall into the class of randomized benchmarking (RB) \cite{KnillBenchmarking,dankert_exact_2009,EmersonRB,PhysRevA.75.022314,MagGamEmer,Mother} (including newer protocols such as \emph{averaged circuit eigenvalue sampling} \cite{flammia2021averaged}).
RB implements suitable sequences of random quantum gates and extracts a
measure of quality as parameters describing the decay rate of the measured signal with the sequence length. This has the advantage of yielding \emph{state preparation and measurement} (SPAM) error robust error metrics.
The experimental sequences of most RB protocols are carefully designed (such as compiled circuit inverses) to efficiently extract specific information from a gate-set.
Prominent exceptions are `filtered' RB protocols such as \emph{linear cross-entropy benchmarking} 
(XEB) \cite{GoogleSupremacy}
that directly work with random sequences of i.i.d.~drawn gates and, e.g., omit an inversion gate.

\begin{figure*}
\includegraphics[width=1\textwidth]{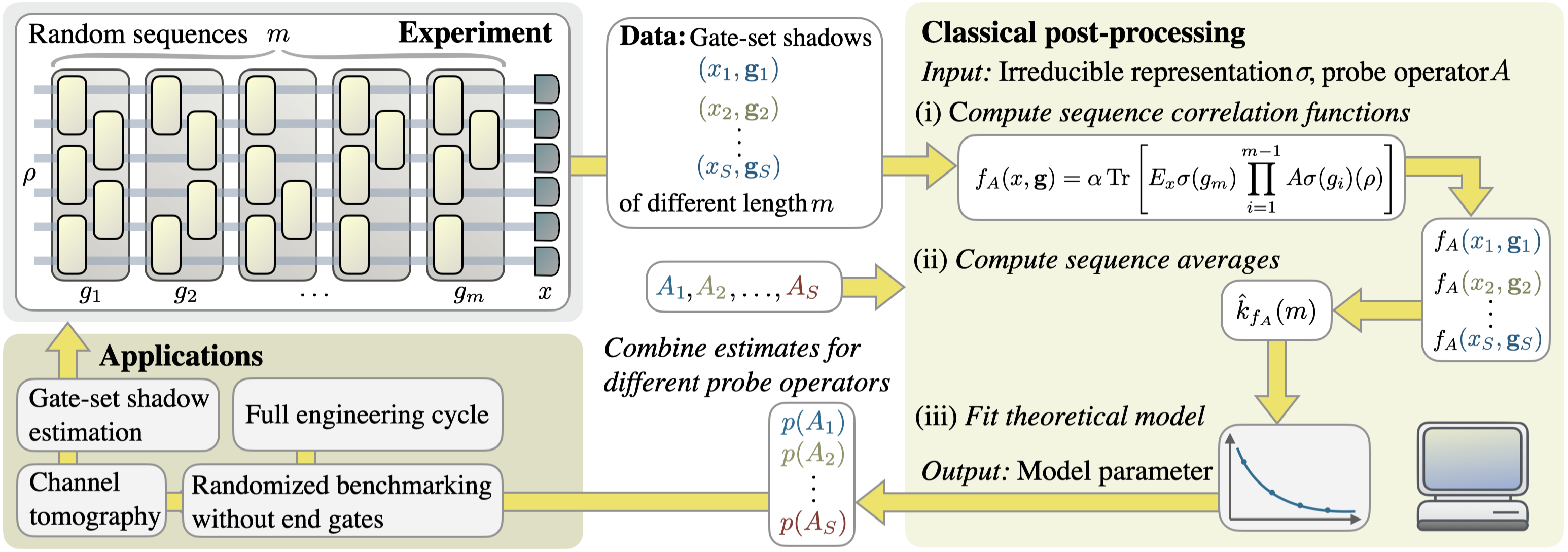}
\caption{
\newtext{The gate-set shadow estimation protocol proceeds in two stages: First, for a fixed inital state $\rho$ and varying sequence lengths $m$ a total of $S$ random sequences of quantum gates of length $m$ are experimentally implemented and each followed by a measurement. We call the observed tuples of measurement outcome and gate sequence $(x^j,g_1^j,\dots,g_m^j)$, $j=1,\dots, S$ the gate-set shadow. The second classical post-processing stage consists itself of three steps: (i) A given sequence correlation function is calculated for every entry of the gate set shadow. For the UIRS protocol a sequence correlation function $f_A$ is specified in terms of a probe super-operator $A$ and an irreducible representation $\sigma$.
(ii) We calculate the the sequence average $\hat k_{f_A} (m)$ as the mean or median-of-means of the result of step (i) over sequences of the same length $m$.
(iii) Sequence averages for different lengths $m$ are used as data points to fit a theoretical model (\cref{eq:uirs_mean}) in order to extract the generalized gate-set fidelity with respect to the super-operator $A$ and the irreducible representation $\sigma$, denoted here by $p(A)$. One of the most important features of this approach is that we can use the same experimental data to accurately estimate \emph{exponentially many} generalized fidelities $p(A_1), p(A_2), \dots, p(A_S)$ by evaluating different sequence correlation functions on the same gate-set shadow. In this way, we can self-consistently and robustly estimate many different properties of the gate-set noise from a minimal amount of data obtained in a simple experiment. Section~\ref{subsec:unitary_noise} to \ref{subsec:channel_reconstruction}  explain and derive guarantees of how the gate-set shadow estimation protocol can be used as a primitive in other more detailed characterization task, such as compressive channel or marginal tomography, potentially allowing one to run the whole engineering cycle on essentially the same type of data.}}\label{workflow_figure}
\end{figure*}
In this work, we take these observations seriously, and revert the mindset that is commonly applied when devising new schemes for benchmarking and characterization. We
ask the question: If all we can feasibly do is implementing unstructured random sequences followed by a native
measurement, what can we learn? At first sight this endeavour is not promising.
Compared to `traditional' RB and tomographic protocols we are giving up on central ingredients.
Thinking about how much information we measure in an unstructured way, we run into the problem that typically, the probabilities of individual measurement results are exponentially small in the number of qubits. This is orthogonal to the careful design of
efficient characterization schemes in prior work, and does not obviously
yield sample efficient estimation schemes at all.

Our change of paradigm is analogous to the mindset of classical shadows~\cite{Shadows, PainiKalev:2019:Shadows}.
Classical state shadows allow for the sample-efficient estimation of (exponentially) many different functions of a quantum state from the same data by only modifying the classical post-processing. Perhaps the central surprise value of the result of ref.~\cite{Shadows} is rigorously guaranteeing that the fidelity of a quantum state with respect to any pure state can be estimated from the same experiment, using only constantly many state copies with sufficiently randomized basis measurements.
This is in stark contrast to schemes like direct fidelity estimation \cite{Flammia-PRL-2011} that given a priori knowledge of the target state carefully optimize the measurements that are performed.

In this work, we define the observed measurement outcomes of random sequences of quantum gates as the \emph{classical shadow of a gate-set} and study the sample efficiency of SPAM-robust estimators for different linear functionals of a gate-set from the same data. Borrowing the median-of-mean estimators used on classical state shadows, we show that the sampling complexity of the estimation (the number of single-shot quantum measurements) can be controlled by a dynamic shadow norm with exponential confidence. We prove bounds on this dynamic shadow norm---a considerably more involved object than its state counterpart---for prominent gate-sets such as the multi-qubit Clifford group and the local Clifford group.
We find that by a suitable post-processing we can estimate the relative average gate fidelities of the noise of a Clifford gate-set with respect to an exponentially large number of unitary channels from polynomially many measurement samples from the same uniformly random experiment.
More generally, we show that the dynamical shadow norm can be controlled in terms of the unitarity of the estimated linear quantity.
Using local gate-sets, we show that one can selectively gain information about channel marginals capturing correlations in their noisy implementation.
We promote this primitive further to design a highly scalable and efficient tomography scheme for cross-talk effects.
Furthermore, we exemplify how gate-set shadows can be used to construct SPAM-robust objective functions for learning noise models and for robust low-rank quantum process tomography.

The important feature of all these schemes is that we only adapt the classical post-processing to the
task at hand, not the quantum experiment.
A single type of data, namely samples from simple local measurements on uniformly random gate sequences, is sufficient to perform a large class of diagnostic tasks of benchmarking, verification and
tomographic recovery.
The mindset can be captured as ``Measure first, ask later!''.
Going beyond uniformly independently random sequences, we can generalize our approach to provide an optimal scheme to learn Pauli noise, emulating the protocol of ref.~\cite{flammia2020efficient} with a simpler experimental
prescription and theoretical analysis.

\textit{Related work.}
We build on a body of literature on randomized schemes for quantum device characterization \cite{Randomized}.
The potential of analyzing the output statistics of gate-set sequences to self-consistently extract essentially all information of a gate-set (as well as the initial state and the measurement) has been realized by gate-set tomography \cite{MerGamSmo13,BluGamNie13,BluGamNie17,Gre15,Nielsen20ProbingQuantumProcessor,Nielsen2020gateSetTomography} with recent variants only requiring random sequences (gate-set shadows) \cite{Gu2021RandomizedLinearGate,Brieger21CompressiveGateSet}.
In contrast to this self-consistent tomographic estimation of all gates in the gate-set, we here target individual linear quantities of the gate-set's average noise or an interleaved quantum process.
Our cross-talk tomography protocol follows the spirit of simultaneous RB \cite{gambetta2012characterization},
but goes significantly beyond simultaneous RB in providing higher-order correlation measures and tomographic information of noise-channel marginals, efficiently from the data of a single randomized experiment.
In~ref.~\cite{KimmOhki}, it has been observed that variants of interleaved multi-qubit Clifford randomized benchmarking experiments \cite{magesan2012efficient} have access to relative average gate fidelities from which unital quantum channels can be reconstructed.
The protocol of ref.~\cite{KimmOhki} performs a different experiment for each fidelity yielding a sub-optimal overall sample complexity for tomography or low-rank tomography \cite{KimLiu16,AverageGateFidelities}.
Gate-set shadow estimation solves both these short-comings.

\section{Results}
\label{sec:results}
\newtext{%
We begin with explaining the general protocol. In the subsequent sections, we then provide theoretical performance guarantees for specific gate-sets and explain how the protocol can be used as a robust estimation primitive in more complex characterization tasks, such as channel tomography.}
\newtext{%
The gate-set shadow estimation protocol consists of two separate stages: an experiment, where measurement results from random circuits of different lengths are recorded, and a classical post-processing step, where different parameters can be estimated from the measured data.}
\newtext{\Cref{workflow_figure}  summarizes the complete protocol.}

\subsection{Protocol: Experiment}\label{experiment}
We aim at characterizing the accuracy of the implementation of a target gate-set $\gr$.
The experimental primitive is the realization of random (gate) sequences of length $m$:
After preparing an initial $\rho$ (e.g., $\ketbra 00$) a sequence of gates $\vec g \in \gr^{\times m}$ is drawn at random according to a distribution  $\mu_m: \gr^{\times m} \to [0,1]$ and applied to $\rho$. This is then followed by a measurement specified by a POVM $\{E_x\}_x$ with measurement outcomes in $\mc X$ (e.g., a computational basis measurement).
If $x \in \mc X$ is observed, the result of the primitive is a tuple $(x, \vec{g}) \in \mc X \times \gr^{\times m}$.

Repeating the primitive multiple times yields a series of tuples $\{(x_i, \vec{g}_i) \}_{i=1}^{S}$ which we refer to as a \emph{(self-consistent) gate-set shadow}. 
\ir{(Note that ref.~\cite{Shadows} actually calls the dual frame elements indexed by the observed output statistics of an informationally complete POVM a state's \emph{shadow}. 
In contrast, we here directly refer to the sampled sequence and observed measurement outcomes as a shadow.)}

A complete experimental protocol further involves measuring such shadows for a set of different sequence lengths $m$.
In order to simplify the theoretical analysis, we focus on the paradigmatic case of $\gr$ being a finite subgroup of $\SU(2^N)$ (such as the Clifford group) and distributions on the sequences arising from the uniform measure over these subgroups.

The simplest example of protocols in this context are \emph{uniform independent random sequence (or \emph{UIRS}) protocols} where the gates in the sequences are drawn from the gate-set uniformly and independently at random. This can be seen as the paradigmatic case, although we will go beyond this later in this work.
We make shadow gate-set estimation through the UIRS protocol explicit for several important gate-sets: namely the multi-qubit Clifford group $\mathbb{C}_n$ and the independent single
-qubit Clifford group $\mathbb{C}_1^{\times n}$ (which we will call the \emph{local Clifford group}).

\subsection{Protocol: Classical post-processing}\label{post}
Given a gate-set shadow $\{(x_i, \vec g_i)\}_{i=1}^S$, we define an empirical estimator in terms of a \emph{sequence correlation function} $f(x,\vec{g}):{\mc{X}}\times\gr^{\times m}\to \mathbb{C}$.
\newtext{For every such sequence correlation function, in the post-processing,
we (i)
evaluate $f$ for all entries of the gate-set shadows and (ii)
calculate the empirical mean or median-of-mean estimator
\begin{equation}\label{km_formula}
  \hat{k}_f (m) \coloneqq \operatorname{(median-of-)means}\ \{f(x_i, \vec g_i)\}_{i=1}^S
\end{equation}
of the result}.
After repeating step (i) and (ii) for different sequence lengths $m$, we fit in step (iii) a theoretical model $k_f$ to the estimates of the \newtext{\emph{sequence means}} $\hat k_f(m)$.
\newtext{After giving this overview of the post-processing protocol, let us take a closer look at the steps and explain their roles in the UIRS protocol:\\}

\newtext{Regarding step (i): Generally speaking,} sequence correlation functions can be seen as the gate-set analog of an observable in shadow estimation.
They allow us to compute properties of noisy gate-sets (for example the average fidelity of an average group element) from experimentally observed gate-set shadows. We emphasize that, like state shadow estimation, the data
collection step of random sequence estimation is independent of the gate-set properties one wishes to estimate, with this estimation step happening entirely in classical post-processing.
Importantly, this enables one to estimate many different correlation functions from the \newtext{same experimental data}.

We here introduce a particular class of sequence correlation functions for UIRS protocols:
Consider an irreducible representation $\sigma$ of $\gr$ with representation space $V_\sigma$.
For the multi-qubit Clifford group, e.g., its adjoint action on traceless Hermitian matrices is of main interest.
We further specify a sequence correlation function in terms of a matrix $A$, POVM $\{E_x\}_{x\in \mc X}$ and state $\rho$, on $V_\sigma$ as
\begin{equation}\label{eq:uirs_correlation_function}
f_A(x,\vec{g}) = \alpha\,\Tr\left[E_x\sigma(g_m)\prod_{i=1}^{m-1} A \sigma(g_i)(\rho)\right]\,,
\end{equation}
with a suitable normalization factor $\alpha$.
\ir{(Note that for $m=1$, and perfectly implemented gates, this expression reproduces the classical state shadows of ref.~\cite{Shadows}.
Generally, restricted to multiplicity-free, irreducible representations, the dual frame construction of ref.~\cite{Shadows} simply amounts to introducing a proper normalization factor, justifying our choice of calling the observed statistics directly the shadow.)}

We refer to $A$ as a \emph{probe (super-)operator} as it specifies the linear quantity of the gate-set that is encoded into the decay parameter of the empirical estimator.
%---
\newtext{Note that the expression \cref{eq:uirs_correlation_function} is closely related to the Born probability of measuring $x$ after applying the sequence $\vec g$ to $\rho$.
The main differences are that we restrict the computation to the subspace $V_\sigma$ and interleave the sequence with the probe operator $A$.
}
%---
\newtext{Similar to classical shadows,
the computation of $f_A$ requires, in general, the same resources as simulating the physical evolution
within a subspace.
In many situations, however, further structure renders this task efficient.
This is in particular the case when both the gate-set and the probe-operators are chosen to be multi-qubit Clifford operations.}

Note that all previously existing RB protocols only use functions that at most depend on the product of the operations in the sequence, $f_{\mathbb{I}}(x, \vec g) = h(x, g_1 g_2 \cdots g_m)$.
In filtered RB protocols, such as
linear cross-entropy benchmarking \cite{GoogleSupremacy}, character benchmarking \cite{helsen2019multiqubit} and Pauli-noise tomography \cite{flammia2020efficient},
 the inversion gate can in this way be omitted and accounted for in post-processing.
Using a non-trivial $A$ goes significantly beyond existing schemes and allows one to even efficiently `interleave in-post' the same data with different probe operators.

\newtext{Regarding step (ii): By taking an empirical average over the gate-set, we expect $\hat k_f(m)$ to be a degree $m$
polynomial in the `average noise' of the gate-set.
One insight of standard Clifford randomized benchmarking is that by taking an uniform average over a sufficiently large group the `average noise' is probed isotropically, effectively projecting it onto a depolarizing channel.
Similarly, UIRS will probe the `average noise' of the gate-set, but by choosing different probe operators $A$, we can alter the operator on which the noise is projected, revealing more information.
Performing the post-processing separately for different irreducible representations $\sigma$, ensures that the gate-set always averages sufficiently over the subspace under consideration.
We will make this intuition precise in the subsequent section.
}

\newtext{Regarding step (iii):
The projection onto isotropic noise (on each representation space) also dramatically `simplifies' the functional form of the expected value of the sequence averages $\hat k_f(m)$.
Recall that for standard Clifford RB, one effectively witnesses a single exponential decay.
Below we show that analogously for UIRS protocols, the theoretical fitting model is a single (matrix) exponential decay encoding linear quantities of the noise in its decay parameter.
The decay parameter(s) can be extracted using least-square fitting algorithms (or tone-finding algorithms such as ESPRIT).
\ir{See ref.~\cite[Sec. VII]{Mother} for a discussion on different post-processing techniques.}
In the end, the UIRS gate-set shadow estimation protocol returns the decay parameters for different choices of probe operators $A$ and representations $\sigma$.
}

\subsection{\ir{Fitting model}}

In order to keep the theoretical derivation and statements concise and straight-forward to interpret, we adhere to some standard assumptions that are commonly used in the analysis of RB protocols.
First, we assume that the quantum channel that implements a sequence $\vec g$ on the quantum device, can be written as $\mc{E}(\vec{g}) = \prod_{i=1}^{m}\phi(g_i)$ with a map $\phi: \gr\to \mc{S}_n$.
Here, $\mc{S}_n$ is the space of $n$-qubit super-operators.
The existence of $\mc{E}$ already excludes, e.g., time-dependent effects in between different experiments, and the factorization into a map $\phi$ further restricts to Markovian noise.
Under this assumption it can be proven that RB protocols \cite{KnillBenchmarking,dankert_exact_2009,EmersonRB,PhysRevA.75.022314,MagGamEmer,Mother} function correctly~\cite{Mother,IndependentNoise,proctor2017randomized}.
For non-Markovian noise much less is known, but in the context of RB rigorous results have been obtained for quasi-static noise \cite{fong2017randomized}, time-dependent noise \cite{RBFlammia} and more recently using tensor-models \cite{figueroa2022towards}. We expect these results to broadly carry over to random sequence estimation.

Second, we assume gate-independent noise, positing the existence of quantum channels $\Lambda_L,\Lambda_R$ such that
\begin{equation}
	\phi(g) = \Lambda_L \omega(g)\Lambda_R %\tag{gate-independent noise}
\end{equation}
where $\omega(g)(\rho) = U_g\rho U_g\ct$ is some ideal
implementation of the gate $g$.
We argue in \cref{subsec:gate-dependent_noise} that our results also apply (up to a negligible error) in the more general Markovian error model, but rigorously proving this (along the lines of ref.~\cite{Mother}) is beyond the scope of this work.

Instead of $\Lambda_R, \Lambda_L$ describing the noise of the gate-set implementation, one can also take the perspective of actively interleaving a channel of interest between a fairly ideal implementation of a gate-set (as is done in interleaved randomized benchmarking \cite{magesan2012efficient}).
While different in protocol and data interpretation, in the analysis, this \emph{black-box query model}
\begin{equation}
	\phi(g) = \Lambda \omega(g) %\tag{black-box channel}
\end{equation}
is simply a special case of the gate-independent noise model and results carry over.

The main analytical result of this work is to establish rigorous performance guarantees for the estimation from gate-set shadows.
The obvious first question being: what do we actually estimate?
As a first result, \newtext{we establish the `simple' model that we should be fitting to the data}.
We show that for a probe operator $A$ the empirical estimator of the protocol converges in probability with the number of samples $S$ in the shadow to a matrix-exponential decay
\begin{equation}\label{eq:uirs_mean}
  \hat k_{f_A}(m) \quad\xrightarrow{S \to \infty}\quad k_{f_A}(m) = \Tr[\Theta \Phi^{m-1}]\,.
\end{equation}
Here, the matrix $\Phi$ depends only on the `between-gates noise channel' $\Lambda:=\Lambda_R\Lambda_L$  and the probe super-operator $A$, while $\Theta$ captures SPAM dependence.
In particular, if $\omega$ contains \newtext{$t$ copies} of the representation $\sigma$ then we have
\begin{equation}\label{eq:phi_matrix}
\Phi_{i,j} = \frac{1}{|P_j|}\tr(P_{i}AP_j \Lambda )
\end{equation}
where $P_i$ is the projector onto the $i$th copy of the representation $\sigma$ inside $\omega$.
Note that here the trace is taken on the space of super-operators.
We give the derivation of this result in \jk{the Supplementary Note 4}.

\newtext{\Cref{eq:uirs_mean} indicates that we should fit a linear combination of (up-to) $t$ exponential decays to the sequence average $\hat k_{f_A}(m)$.
The resulting decay parameters are the eigenvalues of the matrix $\Phi$, which encode information about the overlap of $\Lambda$ and $A$ in the representation space.}

\newtext{A particularly simple fitting model with easily interpretable decay parameters, arises
when the representation $\sigma$ appears in the decomposition of $\omega$ without multiplicities (i.e., there is no other representation in $\omega$ related to $\sigma$ by a change of basis).}
\newtext{If $\sigma$ is multiplicity-free,} then $k_{f_A}(m)$ describes a single scalar exponential decay
\begin{equation}\label{eq:single-decay}
k_{f_A}(m)  \propto p_{\sigma,A}(\Lambda)^{m-1},
\end{equation}
with decay parameter
\begin{equation}
  p_{\sigma,A}(\Lambda) = \frac{1}{d_\sigma}\Tr[ A_\sigma \Lambda]\,
\end{equation}
and $A_\sigma = P_\sigma A P_\sigma$ the probe operator restricted to the representation $\sigma$ of \newtext{dimension $d_\sigma$}.
Note that the proportionality now hides the SPAM-dependent pre-factor.

\newtext{Thus, by fitting a single exponential decay to the empirically observed sequence averages $\hat k_{f_A}$, we can estimate $p_{\sigma, A}(\Lambda)$, the trace-overlap of $\Lambda$ with $A$ on $\sigma$.}
The decay parameter can be thought of as a \emph{generalized fidelity} or effective depolarization
parameter, indicating how much the noise channel $\Lambda$ agrees on average with the probe operator $A$ on the representation space of $\sigma$.

\subsection{\ir{Sample complexity}}
Against the background of the extensively explored variants of RB protocols, the above
\newtext{decay model is not entirely}
unexpected.
A priori less obvious, however, is the sample efficiency of gate-set shadow protocols.
The sequence correlation functions $f(x,\vec g)$ involve normalization factors that typically scale with the dimension of the irreducible representation under consideration.
As a consequence their range can become exponentially large in the number of qubits, causing a simple empirical mean estimator to be susceptible to outliers in the measurement statistics, as well as making a suitably bounded variance a priori nontrivial.
Going significantly beyond the established statistical guarantees in RB, we establish general variance bounds for the UIRS protocol. We do this by introducing a sequence analogue to the shadow norm introduced in ref.~\cite{Shadows} defined on probe super-operator $A$ as opposed to observables. Emphasizing its explicit dependence on the sequence length $m$ we call this norm (really a family of norms indexed by $m$) the \emph{dynamic shadow norm} $\norm{A}_{\mathrm{dyn},m}$. This norm, formally defined in \jk{Supplementary Equation (25)}, depends on the underlying gate-set $\gr$ as well as the ideal input POVM $\{E_x\}$ and state $\rho$. Given these parameters it quantifies the sample complexity of estimating the mean $k_{f_A}(m)$ for arbitrary gate-independent noise. Because of its dependence on the sequence length, the dynamic shadow norm is a more intricate object than its state counterpart. Evaluating it for specific gate-sets accounts for the bulk of the technical innovation in this paper. In terms of the dynamic shadow norm we have the following upper bound on the variance of the UIRS protocol.

\begin{theorem}[Upper bound on the variance]\label{lem:gen_variance}
Consider an UIRS protocol (at sequence length $m$) with gate-set $\gr$ and a
correlation function $f_A$ with probe operator $A$. The variance of the associated mean $k_{f_A}(m)$ is bounded as
\begin{equation}\label{eq:gen_variance}
\mathbb{V}_A(m) \leq \norm{A}_{{\rm dyn},m}\,.
\end{equation}
\end{theorem}
An extended statement and the proof is given \jk{in the Supplementary Note 4}.
The bound on the variance $\mathbb{V}_A(m)$ directly implies a non-asymptotic bound on the sample complexity for the estimator $\hat{k}_{f_A}(m)$ with exponential confidence through the use of median-of-means estimation.
The exponential confidence in particular allows us to estimate `many' quantities
simultaneously from the same shadow data with only logarithmic overhead in the number of quantities. \ir{See Supplementary Note~3 for details.}
More precisely, we get the following guarantee:
Run the UIRS protocol (at sequence length $m$) and measure a gate-set shadow
of $S$ many samples.
Choose a set $\mc A$ of probe operators, an $\epsilon > 0$ and
ensure that for all $A\in \mc A$
\begin{equation}\label{eq:generalSamplingComplexity}
   S \geq C \norm{A}_{{\rm dyn, m}} \frac{\log(|\mc A|)}{\epsilon^2}\,
\end{equation}
for a suitable constant $C$.
Then, in the post-processing,
we obtain $\epsilon$-additive estimates, i.e.,
%\begin{equation},
$
  |{k}_A(m) - \hat{k}_A(m)|\leq \epsilon
$
%\end{equation}
for all $A\in \mc A$.

Hence, bounding the dynamic shadow norm for all $A\in \mc A$ and different sequence lengths $m$ gives simultaneous guarantees for many estimators $\hat{k}_A(m)$ with an overall sampling complexity being the sum of the bounds \Cref{eq:generalSamplingComplexity} for all $m$.
As explained above, $m \mapsto \hat{k}_A(m)$ is then fitted using a theoretical signal model.
For example, in the scenario of multiplicity-free representations giving rise to a single exponential decay \cref{eq:single-decay},
we thereby obtain an estimator for $p_{\sigma, A}(\Lambda)$ for all $A \in \mathcal A$.
The exponential fitting itself is a well-studied problem, for which many advanced techniques~\cite{roy1989esprit,schmidt1986multiple}, flexible software packages~\cite{2020SciPy-NMeth}, and rigorous bounds~\cite{harper2019statistical} can be readily applied.

\subsection{Example: Multi-qubit Clifford UIRS} \label{section:Global_Clifford}
We now provide two particularly practically relevant
examples of UIRS protocols, derive their signal model and a dynamical shadow norm bound guaranteeing their efficiency.

The first example is the multi-qubit Clifford group $\mathbb{C}_n$ that already takes a prominent role in quantum characterization and quantum computation more generally \cite{QEC2}.
We consider an UIRS experiment for $\mathbb{C}_n$: i.e., sequences of i.i.d.~Clifford gates uniformly drawn at random, acting on the initial state
$\ketbra 0 0$ and ending in a computational basis measurement.
This is a common gate-set with a well-understood representation structure, allowing us to explicitly calculate the sequence mean $k_A(m)$ and give bounds on the dynamic shadow norm $\norm{A}_{{\rm dyn},m}$ which controls the sample complexity of sequence estimation.

\paragraph*{Signal model.} The adjoint representation of the multi-qubit Clifford group $\omega(g)$ decomposes into two inequivalent irreducible representations \cite{gross_evenly_2007}: $\sigma_{tr}$ supported on the normalized identity matrix %$\kett{\tau_0}$ a
and $\sigma_{\rm ad}$ supported on the space of traceless matrices, spanned by the generalized Pauli matrices. See \jk{Supplementary Note 2} for details.
We focus on sequence correlation functions with support on $\sigma_{\rm ad}$ only, i.e., $A = P_{\rm ad}AP_{\rm ad}$.
Then, $k_{f_A}(m)$ describes a single exponential decay \cref{eq:single-decay}
with
\begin{equation}
p_{{\rm ad}, A}(\Lambda) = \frac{1}{2^{2n}-1} \tr(P_{\rm ad}AP_{\rm ad}\Lambda)\, .
\end{equation}
This is a familiar quantity: For $A=P_{\rm ad}$, it corresponds to the depolarizing
probability (essentially the average fidelity) of the channel $\Lambda$.
As a very special case the Clifford UIRS protocol in this way emulates standard Clifford randomized benchmarking without performing an inversion.
However, gate-set shadows are considerably more flexible. For instance, by choosing $A = U$ a unitary channel, $p_{{\rm ad}, U}(\Lambda)$ measures the relative average fidelity of $\Lambda$ w.r.t.~the unitary $U$ (i.e., the average fidelity of $U\ad \circ \Lambda$).
In particular, for $U$ a Clifford channel, the corresponding sequence correlation function can be evaluated efficiently.
Relative average gate fidelities are also estimated in interleaved RB.
Compared to existing interleaved RB protocols such as the scheme of ref.~\cite{KimmOhki}, gate-set shadows have the crucial advantage that the experimental protocol itself is independent of $U$.

Since we do not have to implement $A$ on a quantum device, we can also consider $A$ that do not correspond to quantum channels such as rank-one super-operators of the form $X \Tr(Y\, \cdot\,)$for operators $X,Y$. Hence, the gate-set shadows are a versatile tool to estimate properties of the implementation of a Clifford gate-set.

\paragraph*{Dynamical shadow norm.}
The versatility of Clifford UIRS in practice of course crucially depends on the sample efficiency of the estimation.
From the above it is not clear that $k_A(m)$ can be efficiently estimated for arbitrary $A$.
Demanding that $k_A(1) = 1$ in the limit of perfect state preparation, measurement and gates,
the normalization factor $\alpha$ in \cref{eq:uirs_correlation_function} is $\alpha = 2^n+1$,
leading to a single-shot estimator taking values exponentially large in $n$.
Building upon the machinery of the dynamic shadow norm and \cref{lem:gen_variance}, we can still
provide guarantees for efficiently estimatable probe operators and investigate the limits of Clifford UIRS.
As a first step, we assume $A$ to be a restriction of a unitary channel $U$ to the traceless subspace, i.e., $A = P_{\rm ad} UP_{\rm ad}$. In this case, the dynamic shadow norm can in fact be bounded by a small constant independent of the sequence length.

\begin{theorem}[Clifford UIRS unitary norm bound]\label{thm:var_unitary}
For the $n$-qubit Clifford UIRS protocol, $U$ a unitary channel, and $A = P_{\rm ad} UP_{\rm ad} $, it holds that
\begin{equation}\label{eq:var_unitary}
\norm{A}_{{\rm dyn},m} \leq   10\, .
\end{equation}
\end{theorem}

\Cref{thm:var_unitary} is noteworthy for several reasons. First, it does not depend on the number of qubits $n$. Therefore, the estimation of $k_U(m)$ is efficient even on a quantum system consisting of many qubits.
Second, the shadow-norm bound does not depend on the sequence length $m$, enabling relative accuracy estimation of the decay rate in certain regimes. 
We note that the constant $10$ is probably sub-optimal.
The derivation of this theorem can be found \jk{in the Supplementary Note 6.}

As the main consequence of \cref{thm:var_unitary} together with
\cref{eq:generalSamplingComplexity}, we find that it is possible to sample-efficiently estimate exponentially many relative fidelities
with respect to unitary channels to additive precision from the same gate-set shadows obtained by multi-qubit Clifford UIRS.

Next, we consider a general probe super-operator $A$ restricted to the traceless subspace.
Note that $A$ does not need to be a quantum channel.
In the following, we show that the dynamical shadow norm can be controlled in terms of the unitarity
\cite{wallman2015estimating} of $A$,
\begin{equation}
u(A) = \tr(AA\ct)(2^{2n}-1)^{-1}\,.
\end{equation}
For instance, $u(A)\leq 1$ if $A$ is a quantum channel with equality if $A$ is indeed unitary.
We prove the following theorem.

\begin{theorem}[Clifford UIRS general norm bound]\label{thm:cliff_var}
Consider the $n$-qubit Clifford UIRS protocol and let $A = P_{\rm ad} AP_{\rm ad} $ be a probe
super-operator restricted to the traceless subspace.
The dynamic shadow norm for $m > 2$ is upper bounded by
\begin{equation}
\norm{A}_{{\rm dyn},m}\leq  C\,m^2  r(A)^{m-1} \max\{r(A), 1\},
\end{equation}
with
$
r(A) = (1+2^{4-n/3})u(A)
$ and suitable constant $C$.
\end{theorem}
The proof of this theorem, given in the \jk{Supplementary Note 6}, is similar in spirit to \cref{thm:var_unitary}, but significantly more involved.
Choosing $A$ to be unitary $(u(A)=1)$ does not recover \cref{thm:var_unitary}, due to the appearance of the quadratic scaling in $m$.
This term arises because we consider general probe super-operators $A$, giving rise to polynomial transient dynamics in the dynamic shadow norm (due to the non-normality of the underlying operators \cite[Chapter $6$]{wolf2012quantum}).
For many sensible choices of $A$, the polynomial scaling in $m$ does not appear as is evidenced by \cref{thm:var_unitary}. Also, the bound does not quite scale with unitarity $u(A)$, but rather with the parameter $r(A)$ which differs from $u(A)$ by an exponentially small factor. We believe this to be an artifact of the proof technique.

This theorem leads us to the remarkable conclusion that the multi-qubit Clifford UIRS protocol allows us to estimate overlaps $p(A\Lambda)$ for a very large class of super-operators. In particular, $A$ can be any trace non-increasing map, allowing us, e.g., to characterize the overlap between the noise channel $\Lambda$ and sets of Kraus operators, making the Clifford UIRS protocol an all-purpose tool for noise map exploration.

\subsection{Example: Local Clifford UIRS}
A particularly scalable and interesting protocol arises when performing a UIRS protocol with the local Clifford group $\CC_1^{\times n}$ over $n$ qubits.
In this case, the experiment consists of performing sequences of i.i.d.~random single-qubit gates simultaneously on all qubits, initially prepared in $\ketbra 0 0$ ending with a computational basis measurement.

For $\CC_1^{\times n}$ the conjugate representation $\omega(g) = U_g \cdot U_g\ct$ with $U_g = U_{(g_1, \ldots, g_n)} = U_{g_1} \otimes \ldots \otimes U_{g_n}$  decomposes into $2^n$ irreducible, mutually inequivalent representations $\sigma_w$ with $w \in \{0,1\}^n$ that have support on the normalized non-identity Pauli operators on all qubits $i$ for which $w_i=1$.
We denote the projectors onto these irreducible sub-representations as $P_w$ (see \jk{Supplementary Note 2} for more details).

\paragraph*{Signal model.} %
We consider sequence correlation functions with probe operators $A$ that only have support on a single
irreducible representation $\sigma_w(g)$ and set $\alpha = 2^n3^{|w|}$.
Then, the mean $k_{f_A}(m)$ again describes a single exponential decay \cref{eq:single-decay} with

\begin{equation}
p_{w, A}(\Lambda) = \tr(P_w\Lambda P_w A)3^{-|w|}\,.
\end{equation}
We will refer to this quantity as a \emph{local fidelity w.r.t.\ $A$}.
The local fidelity is again somewhat familiar.
The special case $p_{w, \mathbb I}$ has been called the `addressability' in ref.~\cite{gambetta2012characterization}, where it was used to gain information about the strength of correlated errors.
Using gate-set shadows of simultaneously applied local gate sequences, we can collect even more information about correlated errors, giving rise to an efficient \emph{cross-talk tomography protocol} introduced in \cref{subsec:cross-talk_tomography}.
We can again equip the UIRS protocol with sampling complexity guarantees by bounding the shadow norm.

\paragraph*{Dynamic shadow norm.}
We derive a bound on the dynamic shadow norm of the local Clifford group that depends exponentially on the Hamming weight $|w|$ of the bit-string $w$ labeling the representation being addressed but is independent of the total number of qubits in the system.
\begin{theorem}[Local Clifford UIRS norm bound]
\label{thm:local_cliff_var}
For the local Clifford UIRS protocol on $n$ qubits, $w \in \{0,1\}^n$, and $A = P_w A P_w $ a probe operator, it holds that
\begin{align}\label{eq:local_cliff_bound}
\norm{A}_{{\rm dyn},m} &\leq  2^{|w|} 3^{2|w|}\big[3^{-|w|}\tr(AA\ct)\big]^{m-1}\,.
\end{align}
\end{theorem}
The proof is given in \jk{the Supplementary Note 5}.
Note that the term inside the square bracket in \cref{eq:local_cliff_bound} can be considered as a variant of the unitarity restricted to the image of $P_w$. In particular, if $A = P_w U P_w$ for any unitary channel $U$ we have $3^{-|w|}\tr(AA\ct)=1$. Thus, for restrictions of unitary probe operators the bound becomes independent of the sequence length and in consequence the protocol is sample-efficient for bounded $|w|$.

\subsection{\ir{Example beyond UIRS: Pauli-noise estimation}}
Thus far we have focused on uniformly independently sampled random sequences (UIRS protocols).
It is also fruitful to consider more general probability distributions on the set of sequences of a given length.
We give an example for this by constructing a simple protocol that estimates the diagonal elements of an $n$-qubit channel $\Lambda$ using only $O(n2^n)$ samples.
This sampling complexity matches the asymptotic bound given for this task in ref.~\cite{flammia2020efficient}.
Using gate-set shadows, however, gives a simpler experimental description and analysis.
To this end, consider random sequences of the form $\vec{g} = (c^{-1}, p_m , \ldots, p_1, c)$ where $p_1, \ldots, p_m$ are chosen independently uniformly at random from the Pauli group $\mathbb{P}_n$ and $c$ is chosen uniformly at random from the Clifford group $\mathbb{C}_n$.
Note the inverse $c^{-1}$ here at the end of the sequence.
In a black-box fashion, we additionally intersperse the channel $\Lambda$ in between executing the random Pauli elements in the experiment.
The measurement is again a computational basis measurement and the initial state $\rho=\ketbra{0}{0}$.
Choose $\tau$ to be a Hilbert-Schmidt normalized traceless Pauli operator.
As the associated correlation function we define
\begin{equation}
f_\tau(x,\vec{g}) \coloneqq \alpha \Tr[E_x \omega(c)\omega(p_m) A_{\tau}\ldots A_\tau\omega(p_1) \omega(c)\rho]\,
\end{equation}
with $A_\tau \coloneqq \tau \Tr(\tau\,  \cdot)$ %\kett{\tau}\!\braa{\tau}$
and $\alpha = 2^{n}(2^n+1)$.
For convenience, we ignore the SPAM in deriving and stating the following results.
Both of these assumptions can be easily relaxed.
As we show in the \jk{Supplementary Note 7}, the corresponding sequence mean is the power of the diagonal matrix entry of $\Lambda$ corresponding to $\tau$, i.e.,
\begin{equation}
    k_\tau(m) = \Tr[\tau \Lambda(\tau)]^{m-1}.
\end{equation}
We further show that % straightforward calculation (given in \cref{appsec:pauli_var}) shows that
the variance of the associated estimator can be bounded as
\begin{equation}\label{eq:pauli_var}
\mathbb{V}_\tau(m)\leq \frac{2^{3n}(2^n+1)^3}{2^{3n}(2^{2n}-1)} = O(2^n),
\end{equation}
for all choices of $\tau$. Note that there are $4^n-1$ such choices, characterizing all diagonal elements of the quantum channel $\Lambda$. Hence, by using median-of-means estimators, we can estimate $k_\tau(m)$ for all $\tau$ to uniform additive precision using $O(n2^n)$ samples (independently of $m$). By the analysis in ref.~\cite{harper2019statistical} for the estimation of single exponential decays, and the fact that the decay rates $\Lambda_{\tau,\tau}$ are strongly clustered (\cite[Lemma 4]{helsen2019multiqubit}) this leads to a relative precision estimation of the associated Pauli fidelities, matching the performance given in ref.~\cite{flammia2020efficient}.

\subsection{Application: Learning unitary noise models}

\label{subsec:unitary_noise}

In the previous section we have shown how to efficiently estimate the overlap of certain probe-operators with the noise of a gate-set.
This data, e.g., the average gate fidelity of the noise with a specific gate, is already of interest.
The most intriguing feature, however, is that we can estimate many different probe-operators from the same data.
In this way, we can use estimates from gate-set shadows as a subroutine in a complex post-processing pipeline that extracts more information about the noise.
This opens up the way to perform many different characterization tasks that arise in a full-scale engineering cycle of building a quantum computer from the same simple data.
Importantly, the resulting protocols automatically inherit the SPAM robustness of the estimation protocol.
We illustrate these possibilities with three concrete examples.

\begin{figure*}

\includegraphics[width=\textwidth]{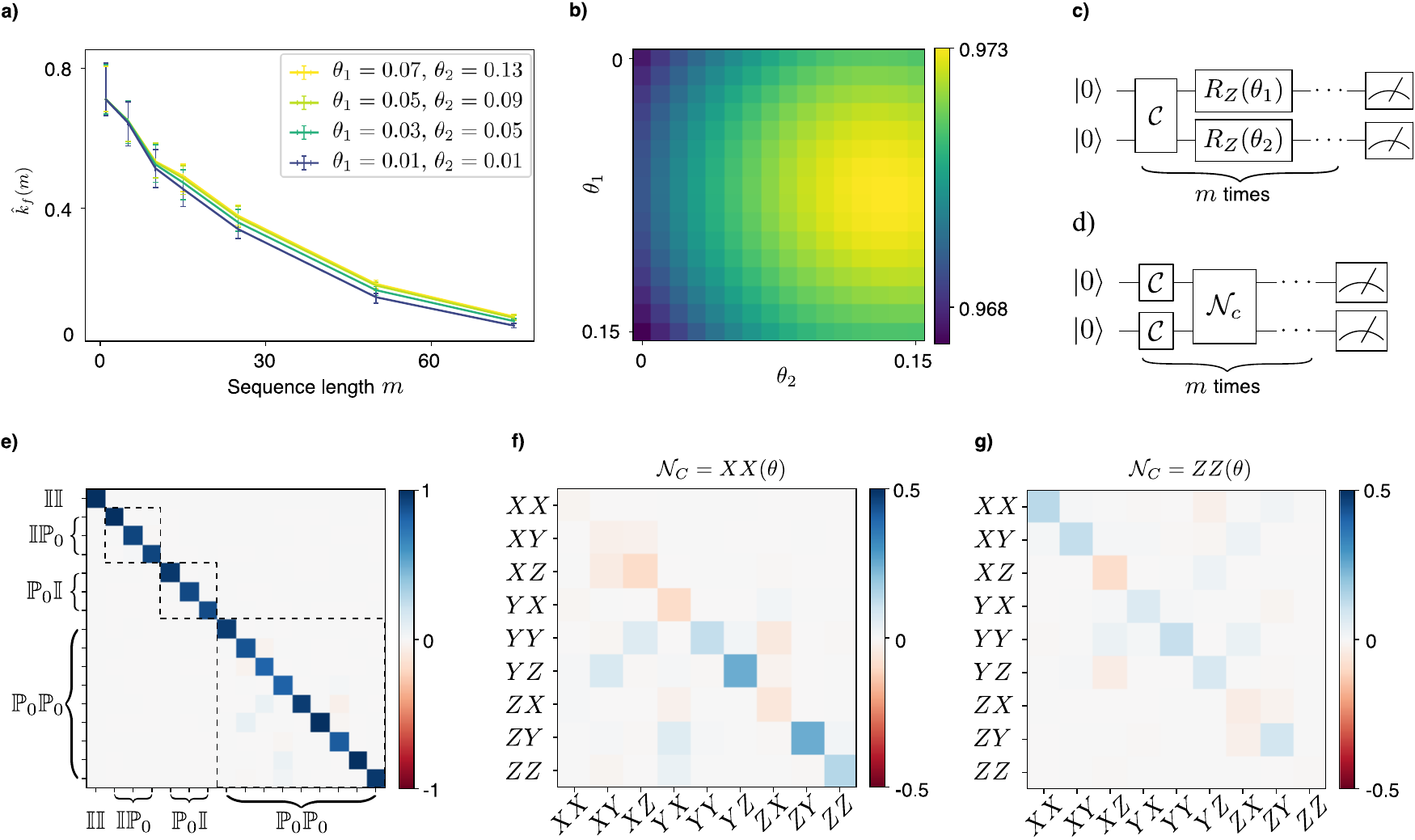}
  \caption{
  Numerical simulations of two potential applications, unitary noise optimization (sec.~\ref{subsec:unitary_noise}) and cross-talk tomography (sec.~\ref{subsec:cross-talk_tomography}).
  Panels \textbf{a)} and \textbf{b)}  show simulation results of the multi-qubit Clifford UIRS protocol for two qubits and 1000 random sequences per sequence length.
  Between every Clifford gate $\mathcal{C}$, two independent $Z$-rotations $R_Z(\theta)$ with rotation angles $\theta_1 = 0.07$ and $\theta_2 = 0.13$ have been applied (see circuit diagram \textbf{c)}).
  Panel \textbf{b)} shows average fidelities $F(U(\theta), \Lambda)$ reconstructed from the gate-set shadows using the ansatz $U(\theta_1, \theta_2) = R_Z(\theta_1) \otimes R_Z(\theta_2)$.
  Example decays of the sequence averages $k(m)$ are shown in panel \textbf{a)} with bootstrapped 95\% confidence intervals around the decay points.
  Panels \textbf{e)} to \textbf{g)} display simulation results for cross-talk tomography from two-qubit local Clifford UIRS data with 15.000 random sequences per sequence length.  After every layer of local Cliffords, an entangling cross-talk noise process $\mathcal{N}_c$
  has been applied (see the circuit diagram \textbf{d)}).
  Panel \textbf{e)} shows the \emph{Pauli transfer matrix} (PTM) of the reconstructed pinched marginal $S$ \jk{Supplementary Equation (107)},
  for cross-talk of the form $\mathcal{N}_c = XX(\theta)$, with dashed boxes indicating the unital marginals $\Lambda_{0,1}$, $\Lambda_{1,0}$, and $\Lambda_{1,1}$.
  Panels \textbf{f)} and \textbf{g)} show the PTMs of the difference between the unital marginal $\Lambda_{1,1}$ and the tensor product $\Lambda_{1,0} \otimes \Lambda_{0,1}$ as a characterization of the cross-talk between the two qubits,
  for cross-talk $\mathcal{N}_c = XX(\theta=0.4)$ in \textbf{f)} and $\mathcal{N}_c = ZZ(\theta=0.4)$ in \textbf{f)}.
  Simulations have been performed using Qiskit \cite{Qiskit} with single-qubit depolarizing noise of $p_1 = 0.002$ for single-qubit gates and two qubit depolarizing noise $p_2 = 0.01$ for two-qubit gates (on top of the custom noise processes after each Clifford layer).
  For the PTM plots, modified functions from the Forest Benchmarking package \cite{forestbenchmarking} have been used. \label{numerics_figure}}
  \end{figure*}

When characterizing noisy quantum gates one differentiates between coherent noise (due to imperfect specification of the gate) and incoherent noise (due to interactions with the environment).
These two types of noise have different consequences for, e.g., error correction \cite{huang2019performance,Roads} and are engineered away in different ways.
At the same time, coherent errors can be corrected by experimental design and control if one has a concrete description.
Given a model for a unitary channel $\theta \mapsto U(\theta)$, we can learn the model parameters $\theta$ approximating the noise channel $\Lambda$ by maximizing $F(U(\theta),\Lambda)$.
During the optimization the objective function, its gradient, etc. can be estimated from the same classical gate-set shadow.
For the multi-qubit Clifford UIRS every estimation requires a polynomial size shadow in the number of qubits and only a logarithmic overhead in the number of evaluations $F(U(\theta),\Lambda)$.
A numerical simulation of a simple learning example is given in \cref{numerics_figure}.

\subsection{Application: Cross-talk tomography}
\label{subsec:cross-talk_tomography}

A key source of error in today's quantum computing devices is correlated noise, or cross-talk.
For this reason, a significant effort has gone into characterizing cross-talk errors specifically \cite{gambetta2012characterization,rudinger2021experimental,Maciejewski2021modelingmitigation}.
Using the flexibility of extracting manifold information from gate-set shadows in the post-processing, we here propose \emph{cross-talk tomography} as an efficient, robust, and detailed cross-talk characterization procedure,
based on the local Clifford UIRS protocol.

The protocol gains tomographic information about, what we call, the \emph{unital marginals} $\Lambda_w = P_w \Lambda P_w$, $w\in \{0,1\}^n$, of the noise channel $\Lambda$.
(Here, $P_w$ is again the projector onto the irreducible representations of the local Clifford group.)
These unital marginals arise as restrictions of channel marginals $\Lambda_{\bar A}$, where one evaluates $\Lambda$ on a maximally mixed input on a system $A$ and traces out the resulting state on $A$ \cite{hsieh2021quantum}.

Now $\Lambda_w$ can be reconstructed via simple linear inversion (see ref.~\cite[Lemma 37]{AverageGateFidelities}) from the local fidelities $p_{w, \mc C}(\Lambda) = 3^{-|w|}\Tr(\Lambda P_w\mc C P_w)$ with respect to the probe-operators given by the local Clifford channel $\mc C$ according to
\begin{equation}
  \Lambda_w = \frac1{|\CC_w|} \sum_{\mc C \in \CC_w} 3^{2|w|}p_{w, \mc C}(\Lambda) \mc C\ad P_w\, ,
\end{equation}
where the sum is restricted to local Clifford channels with unitaries from the subgroup $\CC_w$ of $\CC_n$ acting non-trivially on only the qubits in the support of $w$.
In fact, it is sufficient to consider all local Clifford channels $\mc C$ that act non-trivially on the support of the bit-string $w$.
Not restricting the non-trivial support of $\mc C$, however, allows us to simultaneously reconstruct $\Lambda_w$ for multiple values of $w$.

This constitutes the basis of cross-talk tomography for $k$-local interactions.
Let $H_k \subset \{0,1\}^n$ be the subset of bit strings with Hamming weight $k$.
i) Perform the UIRS experiment for the local Clifford group over $n$ qubits.
ii) Estimate $p_w(\mc C)$ for all $w \in H_k$ and for all $\mc C$ acting non-trivially on the support of $w$.
iii) Reconstruct all $\Lambda_w$ for $w \in H_k$.

By comparing $\Lambda_w$ for different bit strings, one obtains information about the correlations present in $\Lambda$.
Building upon the guarantees for UIRS, we show that cross-talk tomography is $\epsilon$-accurate in diamond norm for all $\Lambda_w$ using $O(k^2\, 2^{9k} / \epsilon^2)$ shadow samples (up-to log-factors).
Thus, for small $k$, cross-talk tomography is highly scalable to large numbers of qubits.
In the light of \cref{thm:local_cliff_var}, this efficiency stems from using local unitary probe operators.
The derivation and even tighter guarantees are given in \jk{Supplementary Note 8}.

As an illustration we study the protocol with a $2$-qubit example. We start by using the local Clifford UIRS protocol to reconstruct the $2$-qubit unital marginals $\Lambda_{1,0}, \Lambda_{0,1}$ and $\Lambda_{1,1}$. Next we compute the tensor product $\Lambda_{1,0}\otimes\Lambda_{01}$. It is straightforward to see that if the channel $\Lambda$ is a tensor product of single-qubit quantum channels featuring no correlations %\footnote{\jk{Drop/merge with main text: Note that this only holds for trace-preserving channels.}} 
(i.e., there is no cross-talk) then  $\Lambda_{1,0}\otimes\Lambda_{0,1} = \Lambda_{1,1}$. Hence, both the difference $\Lambda_{1,0}\otimes\Lambda_{0,1} - \Lambda_{1,1}$ and the product $ \Lambda_{1,1} (\Lambda_{1,0}\otimes\Lambda_{0,1})^{-1}$ provide meaningful characterizations of cross-talk present between qubits $1$ and $2$. The difference measure can be considered as a generalization of the commonly used addressability metric proposed in ref.~\cite{gambetta2012characterization}.
But going beyond a mere metric, we expect that the channel marginals not only detect the presence of cross-talk, but also provide more detailed diagnostic information. As a proof of principle we have numerically simulated the above protocol to diagnose cross-talk in a two-qubit system. The results of a numerical simulation of the protocol are presented in \cref{numerics_figure}.

\subsection{Application: SPAM-robust channel reconstruction}
\label{subsec:channel_reconstruction}

Kimmel et al.~\cite{KimmOhki} have proposed the idea to combine the output of $O(2^{4n})$ different interleaved RB experiments in order to get a robust tomographic estimate of a unital quantum channel $\Lambda$.
By explicitly exploiting the low Kraus-rank, \emph{compressive RB tomography} \cite{KimLiu16,AverageGateFidelities} can reconstruct a unitary approximation to the quantum channel from (up-to-log-factors) $O(2^{2n})$ randomly selected different relative average-gate fidelities with respect to Clifford unitaries.
The previous references, however, left the problem open of providing a SPAM-robust RB protocol that achieves the information-theoretically optimal sampling complexity of $O(2^{4n})$ \cite{AverageGateFidelities} for reconstructing a unitary channel.

We fill in this blank using the data from a multi-qubit Clifford UIRS protocol.
Using a set of randomly selected Clifford unitaries as probe-operators, we can provide the input data to the reconstruction algorithm of ref.~\cite{AverageGateFidelities}.
We show in \jk{Supplementary Note 9} that the number of gate-set shadows to guarantee an accurate reconstruction (in Hilbert-Schmidt norm of the Choi-states) indeed matches the lower-bound of $O(2^{4n})$.
Note that the number of channel invocations is bounded by the maximal sequence length times the number of sequences.
Besides the favorable scaling, the UIRS protocol has the crucial advantage compared to, e.g., the interleaved protocol of ref.~\cite{KimmOhki} that the same measurement data is used for estimating all the average fidelities.

Going beyond the compressive reconstruction of unitary quantum channels, we can use Clifford UIRS as a primitive for the robust reconstruction of arbitrary unital quantum channels in the spirit of ref.~\cite{KimmOhki}, see also ref.~\cite[theorem 38]{AverageGateFidelities} and ref.~\cite{Sco08}.
The required size of the gate-set shadow is $O(2^{8n})$ for an accurate reconstruction in any norm in which unitary channels are normalized.

\subsection{Gate-dependent noise}\label{subsec:gate-dependent_noise}
The presentation so far assumed gate-independent noise. This assumption can be substantially relaxed, at the cost of introducing a more complex description of the noise.
We will focus on the UIRS protocol, which is particularly robust against gate-dependent fluctuations.
We give a fairly comprehensive argument, but leave a rigorous proof of the robustness to future work.
Our argument follows that of the robustness against gate-dependent errors for RB~\cite{IndependentNoise,Mother}. 
For gate-dependent noise, the data form in expectation can be generally written as
\begin{equation}
  k_A(m) = \tr\left[\Xi\, (A \otimes \mathbb I)(\mc{F}(\phi)[\sigma])^{m-1}\right],
\end{equation}
where $\Xi$ depends on the state and measurement and the
operator
$\mc{F}(\phi)[\sigma] \coloneqq \EE_{g \in \gr}\ \sigma(g)\otimes \phi(g)$
is known as the (non-commutative) Fourier transform of $\phi$, evaluated at the irreducible representation $\sigma$;
see the derivation of \jk{Theorem 7 in the Supplementary Information}.

A key fact about this Fourier transform (see, e.g., ref.~\cite{merkel2018randomized} for a proof) is that if $\phi$ is a representation $\omega$ (i.e., a perfectly implemented gate-set), then $\mc{F}(\phi)[\sigma]$ is an orthogonal projector with rank equal to the number of copies of $\sigma$ present in $\omega$. For simplicity, let $\omega$ be multiplicity-free. Then, $\mc{F}(\phi)[\sigma]$ is a rank-one projector.
This implies that $(A\otimes \mathbb{I}) \mc{F}(\phi)[\sigma]$ %= \kett{v(AP_\sigma)}\braa{v(P_\sigma)}$
is also a rank-one projector.
When $\phi$ is a sufficiently `good' implementation of $\omega$, the difference between $\mc{F}(\phi)[\sigma]$ and $\mc{F}(\omega)[\sigma]$ is small (in some suitable norm) and can be regarded as a perturbation of $\mc{F}(\omega)[\sigma]$.
(See ref.~\cite{Mother} for a discussion of norms on this space.)
Applying the perturbation theory of non-normal matrices, we conclude that $(A\otimes \mathbb{I})\mc{F}(\phi)[\sigma]$ is as well approximately rank-one, and in particular that there exist super-operators $\Lambda_L, \Lambda_R$ such that
\begin{equation}
\big( (A\otimes \mathbb{I})\mc{F}(\phi)[\sigma]\big)^{m} = \left(A \Lambda_R P_\sigma \tr[P_\sigma\Lambda_L\,\cdot\,]\right)^m + E^m
\end{equation}
where $E$ is a matrix of small norm and $P_\sigma$ is the projector onto $\sigma$ (in the image of $\omega$).
This means that the decay rate $k_A(m)$ has the general functional form
\begin{equation}
k_A(m) = B_1p(A)^{m-1} +B_2\delta(E)^{m-1}
\end{equation}
where $B_1,B_2$ are real numbers encoding SPAM, $\delta(E)$ is small, and $p(A)$, the dominant eigenvalue of $(A\otimes \mathbb{I})\mc{F}(\phi)[\sigma]$, is given by
\begin{equation}
p(A) = {|P_\sigma|^{-1}}\tr(\Lambda_L  P_\sigma A P_\sigma \Lambda_R )\,.
\end{equation}
Up to a small and exponentially decreasing error, we thus recover the functional form of \cref{eq:uirs_mean} also in the presence of gate-dependent noise.
It is important to note however, that in this general case $\Lambda_L$ and $\Lambda_R$ (and their product) need not be CPTP. This complicates the interpretation of $p(A)$ as describing an aspect of a physical noise process.

\section{Discussion}\label{sec:outlook}

It has long been known that classical randomness can facilitate the construction of informative characterization protocols for quantum devices.
Randomized benchmarking~\cite{KnillBenchmarking,dankert_exact_2009,EmersonRB,PhysRevA.75.022314,MagGamEmer,Mother} and
classical shadow estimation~\cite{Shadows, PainiKalev:2019:Shadows} are examples of this mindset.
In our work, we follow this paradigm even more stringently for diagnosing noise in gate-set implementations.
Instead of engineering sophisticated and specific experimental protocols for a specific task, we turn the approach upside down: we focus on the `simplest' randomized protocol that can be implemented with current and near-term quantum architectures: Random gate-sequences followed by native measurements.
Accepting this restriction, we then ask how detailed diagnostic information can be extracted from the resulting data and most importantly how many samples are required.

It turns out that the resulting prescription---a single experiment that
can and has been implemented experimentally already---allows for solving many
benchmarking, certification and identification problems with (near-)optimal efficiency.
All the technicalities that come along with different tasks are shifted
to the classical post-processing phase.
Most importantly, multiple diagnostic tasks can be performed from the same measurements, allowing us to base an entire engineering cycle on a single experiment.

The ideas advocated here constitute the beginning rather than the conclusion of a program.
We regard our theoretical results as a strong motivation to experimentally realize and make use of the concrete applications, such as robust learning of unitary noise and cross-talk tomography.
In addition, several further extensions seem exciting.
A logical first extension of our work is UIRS with other groups and non-uniform measures over said groups.
As with state shadow tomography and randomized benchmarking, we believe the UIRS protocol can be furnished with rigorous guarantees for several other useful gate-sets such as the matchgates \cite{helsen2020matchgate,zhao2021fermionic}, the Heisenberg-Weyl group, the ${\rm CNOT}$-dihedral group and even gate-sets that do not constitute a group \cite{proctor2019direct,liu2021benchmarking}.

We also illustrated the potential of using correlated sequences where the gates are not drawn independently.
We believe that using simple correlated sequences gives a fruitful perspective on long-standing problems such as the characterization of non-Markovian and time-varying noise processes in an experimentally friendly and scalable way.
Furthermore, while not demonstrated here, akin to their state analog, gate-set shadows can also be used for estimating non-linear quantities.

While the bulk of this work discusses diagnostic tools for developing near-term quantum computing devices, random sequence protocols apply beyond that.
We expect that gate-set shadows will for instance find application as a primitive in quantum machine learning \cite{EfficientML}, in particular in dynamic settings such as time-series estimation.
Also in this context, the possibility to `measure first and ask later' increases the flexibility in devising hybrid quantum-classical schemes with experimentally feasible quantum computations. \\

\section{Data availability}

\newtext{
The simulated data used for creating the plots in \cref{numerics_figure} have been deposited on Figshare and are publicly
available \cite{figshare}.}

\section{Code availability}

\newtext{The code used to simulate the protocol and create the plots in \cref{numerics_figure} is available upon request.}

\begin{widetext}
\newpage

\begin{center}
{\large Supplemental material}
\end{center}
\section{Notation}

Throughout this work, we are operating in the \emph{Liouville} or \emph{transfer matrix representation} of quantum channels. We represent finite-dimensional $d\times d$ density matrices $\rho$ as length $d^2$ column vectors $\kett{\rho}$ and POVM elements $E$ as length $d^2$ row vectors $\braa{E}$, with a corresponding trace-inner product $\braa{E}\rho\rangle\!\rangle = \tr(E\ct \rho)$.
In this picture, super-operators $\mc{E}$ get mapped to $d^2\times d^2$ matrices $\mc{E}$ with the property $\mc{E}\kett{\rho} = \kett{\mc{E}(\rho)}$. Note that this representation is compatible with composition of super-operators (mapping to matrix multiplication), and the taking of tensor products. When $d$ is a power of two a good basis for the space of matrices is the set of Hermitian Pauli operators $\mathbb{P}^*$ (normalized under the trace inner product), in this case we also always write $d=2^n$. We denote the normalized identity by $2^{-n/2}\mathbb{I} = \tau_0$ and the set of normalized traceless Hermitian Pauli matrices $\tau$ as $\mathbb{P}^*_0$.
Finally, we use a tilde to indicate noisy implementations of POVM elements and states, so $\tilde \rho$ is a noisy implementation of the state $\rho$ and $\{\tilde E_x\}_x$ is a noisy implementation of the POVM $\{E_x\}_x$. For the specific case of the all-zero computational basis state $\dens{0}\tn{n}$ we write $\kett{0_n}$ (with noisy version $\kett{\tilde 0_n}$) and for the computational basis POVM $\dens{x}$ we write $\braa{x}$ (with noisy version $\braa{\tilde x}$).

\section{Technical preliminaries on representation theory}\label{appsec:prelim}
In this section we recall some basic facts of representation theory (of finite groups), and discuss generally how it applies to our work, with a particular focus on the representation theory of the Clifford group. For a more in depth introduction to representation theory,
we recommend the standard textbook ref.~\cite{fulton2013representation}.

%\subsection{Representations}
Let $\gr$ be a finite group and consider the space $\mc{M}_d$ of linear transformations of $\mathbb{C}^d$. A representation $\omega$ is a map $\omega:\gr \rightarrow \mc{M}_d$ that preserves the group multiplication, i.e.,
\begin{equation}
\omega(g)\omega(h) = \omega(gh),\;\;\;\;\;\; \forall g,h \in \gr.
\end{equation}
We will require the operators $\omega(g)$ to be unitary as well (for finite groups this can always be done).

\paragraph*{Reducible and irreducible representations.}
If there is a non-trivial subspace $W$ of $\mathbb{C}^d$ such that for all vectors $w\in W$ we have
\begin{equation}\label{eq:sub-rep}
\omega(g)w\in W,\;\;\;\;\;\;\forall g\in \gr,
\end{equation}
then the representation $\omega$ is called \emph{reducible}. The restriction of $\omega$ to the subspace $W$ is also a representation, which we call a \emph{sub-representation} of $\omega$. If there are no non-trivial subspaces $W$  such that \cref{eq:sub-rep} holds the representation $\omega$ is called \emph{irreducible}. We will generally reserve the letter $\sigma$ to denote irreducible representations.
Two representations $\omega,\omega'$ of a group $\gr$ are called \emph{equivalent} if there exists an invertible linear map $T$ such that
\begin{equation}\label{eq:equivalent}
T\omega(g) = \omega'(g)T,\;\;\;\;\;\;\; \forall g\in \gr.
\end{equation}
We will denote this by $\omega\simeq \omega'$.

\paragraph*{Sums, products, and Maschke's Lemma.}
We will make use of sums and products of representations.
Given representations $\omega, \omega'$, the maps
\begin{align}
&\omega\oplus\omega':\gr \to \mc{M}_d \oplus \mc{M}_{d'}:g \mapsto \omega(g)\oplus \omega'(g),\\
&\omega\otimes\omega':\gr \to \mc{M}_d \otimes \mc{M}_{d'}:g \mapsto \omega(g)\otimes \omega'(g),
\end{align}
are again representations. They are, however, generally not irreducible (even if $\omega$ and $\omega'$ are).
However, Maschke's Lemma ensures that every representation $\omega$ of a
group can be uniquely written as a direct sum of irreducible representations, that is
\begin{equation}\label{eq:Maschke_Lemma}
\omega(g) \simeq \bigoplus_{\lambda\in S} \sigma_\lambda(g)^{\oplus n_\lambda},\;\;\;\;\;\;\forall g\in \gr,
\end{equation}
where the index set $S$ labels a subset of the irreducible representations of $\gr$ and $n_\lambda$ is an integer denoting the number of copies (or multiplicity) of $\sigma_\lambda$ present in $\omega$.

\paragraph*{Averages of representations.}
Here we recall some standard results about averages over representations of finite groups. We will present these without proof, referring again to ref.~\cite{fulton2013representation} for a more detailed explanation.
First is the basic statement that the average over any representation of a finite group is a projector (precisely onto the subspace on which the representation acts trivially):

\begin{lemma}\label{lem:projector}
Let $\omega$ be a representation of a group $\gr$ then
\begin{equation}
\avg_{g\in \gr} \omega(g) = P_{\rm inv}
\end{equation}
where $P_{\rm inv}$ is the projector onto the subspace left invariant under the action of $\omega(g)$, i.e., all vectors $v$ s.t. $\omega(g)v = v$ for all $g\in \gr$.
\end{lemma}

Second is a useful statement about the invariant subspaces of two-fold tensor powers of representations.
\begin{lemma}\label{lem:square}
Let $\sigma,\sigma'$ be real, irreducible, inequivalent, and non-trivial representations of a finite group $\gr$. Then the representation $\sigma\otimes \sigma'$ has no invariant subspace, while the representation $\sigma\tn{2}$ leaves the vector $v(P_\sigma)$ invariant, where $v(P_\sigma) $
is the vectorized projector onto the image of $\sigma$.
\end{lemma}

\paragraph*{Representation theory of the Clifford group.}
Here we give some basic facts about the representation theory of the Clifford group, which are used in the main text to derive the decay models for the multi-qubit and local Clifford UIRS protocols.

\begin{lemma}\label{lem:multi_reps}
The Liouville representation of the $n$-qubit Clifford group $\mathbb{C}_n$ decomposes into two irreducible representations, in particular we have for all $g \in \mathbb{C}_n$:

\begin{equation}
\omega(g)=\omega_{\rm triv}(g)\oplus\omega_{\rm ad}(g)
\end{equation}
where $\omega_{\rm triv}(g)$ has support on $\mathrm{Span}\{\tau_0\}$ and $\omega_{\rm ad}(g)$ has support on the space of traceless matrices spanned by all normalized traceless Hermitian Pauli operators $\tau\in \mathbb{P}^*_0$ .
\end{lemma}
This is a direct consequence of the $2$-design property of the Clifford group. An early proof can be found in ref.~\cite{gross_evenly_2007}. A direct consequence is the following.
\begin{lemma} \label{lem:multi_proj}
Let $\omega$ be the Liouville representation of the $n$-qubit Clifford group. Then we have
\begin{equation}
\avg_{g\in \mathbb{C}_n} \omega(g)^{\otimes 2}= \kett{\tau_0^{\otimes2}}\braa{\tau_0^{\otimes2}}+\frac{1}{2^{2n}-1} \sum_{\tau,\tau' \in \mathbb{P}^*_0}\kett{\tau\tn{2}}\!\braa{{\tau'} \tn{2}}.
\end{equation}
\end{lemma}

We have similar statements for the local Clifford group.
\begin{lemma}
The Liouville representation of the local Clifford group $\mathbb{C}_1^{\times n}$ on $n$ qubits decomposes into $2^n$ mutually inequivalent irreducible representations
\begin{equation}
\omega(g)=\bigoplus_{w \in \{0,1\}^n} \sigma_w(g)
\end{equation}
where $\sigma_w(g)$ has support on $\mathrm{Span}\{ \mathbb{P}^*_w\}$ with
\begin{equation}
\mathbb{P}_w = \Big\{\bigotimes_{i=1}^n\tau_{i} \;\;\;\|\;\;\; \tau_{i}= \tau_0 \text{  if     }\; w_i= 0 \text{   and    } \tau_i \in \{\tau_X,\tau_Y, \tau_Z\} \text{  if   }\; w_i= 1\Big\}.
\end{equation}
\end{lemma}
This is a direct result of the previous lemmas, applied to each of the $n$ qubits individually.
We also have the following statement.

\begin{lemma}\label{second_moment_local_clifford}
Let $\omega$ be the Liouville representation of the local Clifford group on $n$ qubits. Then we have
\begin{equation}
\avg_{g\in \mathbb{C}_1^{\times n}} \omega(g)^{\otimes2}=\sum_{w\in \{0,1\}^n} \frac{1}{3^{|w|}}\bigg(\sum_{\tau\in \mathbb{P}^*_w}\kett{\tau\tn{2}}\bigg)\bigg(\sum_{\tau'\in \mathbb{P}^*_w}\braa{{\tau'}\tn{2}}\bigg),
\end{equation}
where again
\begin{equation}
\mathbb{P}^*_w = \Big\{\bigotimes_{i=1}^n\tau_{i} \;\;\;\|\;\;\; \tau_{i}= \tau_0 \text{  if  }   \; w_i= 0 \text{   and    } \tau_i \in \{\tau_X,\tau_Y, \tau_Z\} \text{  if   }\; w_i= 1\Big\}.
\end{equation}
\end{lemma}

\section{Simultaneously estimating many observables}\label{ssec:est_obs}

Key to the results in this work is the following general statistics observation, which also powers state shadow estimation \cite{Shadows}. Let $p$ be a probability distribution over some (finite) set $\mc{X}$, and let $\mathcal A$ be a set of observables, i.e., functions $f:\mc{X}\to \mathbb{R}$. Suppose we wish to estimate the vector of means $[\mathbb{E}_{\mc{X}} (f_A)]_{A\in \mathcal A}$ to some overall error $\epsilon$. A surprising fact from mathematical statistics is that this is possible by drawing only $O(\log(|\mathcal A|)\mathbb{V}_{\rm max}(f_A)/\epsilon^2)$ samples from $p$ where $\mathbb{V}_{\rm max}(f_A) = \max_{A\in \mathcal A}\mathbb{V}_X(f_A)$ is the maximal variance of the functions in $\mathcal A$. Doing this (without making strong assumptions on the observables $f_A$) requires the construction of so-called sub-Gaussian estimators (see refs.~\cite{SubGaussian,MedianOfMeans,lugosi2019mean}
for reviews)
for the means $\mathbb{E}_{\mc{X}} (f_A)$.
An example of such an estimator that is straightforward to compute is the \emph{median-of-means estimator}, which
has been used in state shadow estimation by ref.~\cite{Shadows}. Following their notation, it involves gathering $S = NK$ samples $\{x_i\}_{i=1}^{NK}$ from the distribution $p$,
where $N, K$ are integers. For an observable $f_A$, one can then construct the estimator
\begin{equation}
\hat{f}_A = \!\text{Median}\bigg\{\!\frac{1}{N}\!\!\!\sum_{i=I}^{I+N-1}\!\!\!\!f_A(x_i) \;\|\; I\in \{1,\! N+1,\! 2N+1,\ldots ,\! (K-1)N + 1\}\!\bigg\}
\end{equation}
for the average $\mathbb{E}(f_A)$, splitting the data into $K$ equally sized parts of size $N$.
It can be shown that if we set
\begin{align}\label{eq:sample_comp}
K &= \big\lceil 2 \log(2|\mathcal A|/\delta)\big\rceil ,\\
N &= \bigg\lceil\frac{34}{\epsilon^2}\mathbb{V}_{\rm max}(f_A)\bigg\rceil,
\end{align}
then we have
\begin{equation}
\max_{A\in \mathcal A}|\hat{f_A}-\mathbb{E}(f_A)|\leq \epsilon,
\end{equation}
with probability $1-\delta$. We can substitute $\epsilon$ to obtain the direct relation
\begin{equation}\label{eq:max_var}
\max_{A\in \mathcal A}|\hat{f_A}-\mathbb{E}(f_A)|\leq  \sqrt{\frac{68\mathbb{V}_{\rm max}(f_A) \log(2|\mathcal A|/\delta)}{NK}},
\end{equation}
in terms of the total number of samples $NK$.
Hence, providing bounds on the maximal variance of a set of observables provides a rigorous guarantee on their estimation at any degree of confidence. Note however that the construction of the estimator is dependent on the level of confidence $\delta$ (through the setting of $K$). This is unfortunate, but it turns out to be impossible \cite{lugosi2019mean} to drop this requirement for sub-Gaussian estimators.

\section{Guarantees for the UIRS protocol}\label{appsec:uirs}
In this section, we give the derivations of the general performance guarantees for the UIRS protocol summarized in the main text.
\subsection{Fitting model}
As we have argued in the main text, a useful class of sequence correlation functions is given by
\begin{equation}\label{eq:corr_func_app}
f_A(x,\vec{g}) = \alpha\braa{E_x}\sigma(g_m)\prod_{i=1}^{m-1} A \sigma(g_i))\kett{\rho}\,,
\end{equation}
where $A$ is some fixed \emph{probe super-operator}, $\alpha$ is a suitable normalization and $\sigma(g),\phi(g)$ are representations of the gate-set group $\gr$. We begin by deriving the main result (eq. (5) in the main text) on the mean $k_A(m)$ in the UIRS protocol.

\begin{theorem}\label{thm:mean_signal_form}
{Let $k_A(m)$ be the outcome of an UIRS experiment with a correlation function as in \cref{eq:corr_func_app}, over a gate-set $\gr$. Then we have, under the assumption of gate-independent noise,}
\begin{equation}
k_A(m) = \tr\big( \Theta(\{E_x\}_x, \rho)[\Phi(A,\Lambda)]^{m-1}\big),
\end{equation}
where $\Theta ,\Phi$ are matrices induced by the representation structure of $\omega(g)$. $\Phi(A,\Lambda)$ depends only on the between-gates noise channel $\Lambda:= \Lambda_R\Lambda_L$ and the probe super-operator $A$, while $\Theta$ captures state preparation and measurement (SPAM) dependence. In particular, if $\omega$ contains $n_\sigma$ copies of the representation $\sigma$ then we have
\begin{equation}
\Phi_{i,j} = \frac{1}{|P_j|}\tr(P_{i}AP_j \Lambda ),
\end{equation}
where $P_i$ is the projector onto the $i$th copy of the representation $\sigma$ inside $\omega$.
\end{theorem}
The formulation of the result given in the main text directly follows from \cref{thm:mean_signal_form} by additionally realizing that $f_A(x,\mathbf g)$ regarded as a random variable pushing forward $p(x, \mathbf g)$ is bounded and, thus, the corresponding mean estimator $\hat k_{f_A}(m)$ is unbiased and consistent.
Correspondingly the median-of-mean estimator converges to the expected value of the mean.

\begin{proof}
Recall that $\phi(g)=\Lambda_L \omega(g) \Lambda_R$. Hence
\begin{align}
k_A(m)&=\avg_{\vec{g}\in\gr^{\times m} } \sum_{x\in \{0,1\}^n}\alpha\braa{E_x}\sigma(g_m)\prod_{i=1}^{m-1} A \sigma(g_i))\kett{\rho} \braa{\tilde{E}_x} \prod_{i=1}^m \Lambda_L \omega(g_i) \Lambda_R \kett{\tilde{\rho}} \\
&=\avg_{\vec{g}\in\gr^{\times m} } \sum_{x\in \{0,1\}^n} \alpha \tr
\bigg(
\big(
\kett{\rho\otimes \Lambda_R(\tilde\rho)}\braa{E_x\otimes \Lambda_L^*(\tilde{E}_x)}
\big)
\sigma(g_m)\otimes \omega(g_m)\prod_{i=1}^{m-1}
\big(
(A  \otimes \Lambda)(\sigma(g_i)\otimes \omega(g_i))
\big)
\bigg)\notag\\
&=\sum_{x\in \{0,1\}^n} \alpha \tr
\bigg[
\big(
\kett{\rho \otimes \Lambda_R(\tilde\rho)}\braa{E_x \otimes \Lambda_L^*(\tilde{E}_x)}
\big)
\left( \avg_{g\in\gr }(\sigma(g)\otimes \omega(g)) (A  \otimes \Lambda) \avg_{g\in\gr }(\sigma(g)\otimes \omega(g))
\right)^{m-1}
 \bigg)\bigg],\notag
\end{align}
where we have used that the representation average is a projector (and thus equal to its square).
Now note that we can write $\omega(g) = \sigma^{n_\sigma}(g) \oplus \omega'(g)$ where $\omega'$ is a representation that contains no copies of $\sigma$. From this and \cref{lem:projector,lem:square} given above we can see that
\begin{align}
k_A(m)&=\sum_{x\in \{0,1\}^n}\alpha\tr
\bigg(
\big(
\kett{\rho \otimes\Lambda_R(\tilde\rho)}\braa{E_x \otimes \Lambda_L^*(\tilde{E}_x)}
\big)
\left (
\sum_{i,i' = 1}^{n_\sigma} \frac{1}{d_\sigma}\kett{v(P_i)}\braa{v(P_{i'})} (A  \otimes \Lambda)\sum_{j,j' = 1}^{n_\sigma} \kett{v(P_j)}\braa{v(P_{j'})}
\right)^{m-1}
\bigg),
\end{align}
where $P_i$ is the projector on the $i$'th copy of $\sigma$ in $\omega$. Using the fact that
\begin{equation}
\braa{v(P)} A\otimes B \kett{v(P')} = \tr(A^T P B P'),
\end{equation}
and defining the matrices $\Theta, \Phi$  appropriately, we obtain the theorem statement.
\end{proof}

\subsection{Variance bound with the dynamic shadow norm}
In order to bound the sampling complexity of the estimation of UIRS means $k_A(m)$ it is sufficient, through the use of median-of-means estimators, to obtain a bound on the variance of associated probability distribution. In the main text we did this by introducing the dynamic shadow norm.
The dynamic shadow norm is formally defined as
\begin{align}\label{eq:shadow_norm}
\norm{A}_{{\rm dyn},m} &= \alpha^2\max_{\Lambda_R, \Lambda_L}\bigg|\sum_{x\in \mc{X}}\!\braa{E_x\tn{2} \!\otimes\!\Lambda_L^*(\tilde E_x)}\big(P_\sigma^{(2)}(A\tn{2}\!\otimes \!I) P_\sigma^{(2)} \big)^{m\!-\!1}\kett{\rho\tn{2}\!\otimes\! \Lambda_R(\tilde\rho)}\bigg|,
\end{align}
with
\begin{equation}
P_\sigma^{(2)} = \avg_{g\in \gr} \sigma (g)^{\otimes 2}\otimes \omega(g).
\end{equation}

We prove the associated theorem:

\begin{theorem}[Restatement of Theorem 1 in the main text]
  \label{thm:uirs_variance_bound}
Consider an UIRS protocol  (at sequence length $m$) with gate-set $\gr$. Also consider a correlation function $f_A$ with probe super-operator $A$. The (single-shot) variance of the associated mean estimator $\hat{k}_{f_A}(m)$ is bounded as
\begin{equation}
\mathbb{V}_A(m) \leq \norm{A}_{{\rm dyn},m}.
\end{equation}

\end{theorem}

\begin{proof}
The variance of a discrete random variable $X$ assuming values $x\in {\mc{X}}$ with probability $p(x)$ is given by
\begin{equation}
\mathbb{V}(X)= \sum_{x\in \mc{X}} (x-\mu)^2p(x),
\end{equation}
with $\mu$ the expected value of $X$. We can therefore obtain an upper bound on the variance by simply considering
\begin{equation}
\mathbb{V}(X)\leq  \sum_{x\in \mc{X}} (x)^2p(x).
\end{equation}
Thus, we have that
\begin{align}
\mathbb{V}_A(m) = \mathbb{V}(X_A(m))\leq \alpha^2\avg_{\vec{g}\in \gr^{\times m}}\sum_{x\in \mc{X}} \braa{E_x}\sigma(g_m)\prod_{i=1}^{m-1} A \sigma(g_i))\kett{\rho}^2 \braa{\tilde E_x} \prod_{i=1}^m \phi(g_i) \kett{\tilde \rho}.
\end{align}
Using the identities $\tr(A\otimes B)=\tr(A)\tr(B)$ and $AB\otimes AB=A^{\otimes 2} B^{\otimes 2}$ we obtain
\begin{align}
\mathbb{V}_A(m) \leq \alpha^2\sum_{x\in \mc{X}} \tr\bigg( \big(\kett{\rho\tn{2}\otimes \Lambda_R(\tilde\rho)}\braa{E_x\tn{2}\otimes\Lambda_L^*(\tilde E_x)}\big)\avg_{g\in \gr} \sigma (g)^{\otimes 2}\otimes \omega(g)\bigg( \avg_{g\in \gr} (A\sigma(g))\tn{2} \otimes (\Lambda\omega(g)) \bigg)^{m-1} \bigg).
\end{align}
Maximizing over $\Lambda_R, \Lambda_L$ and recalling the definition
\begin{align}
\norm{A}_{{\rm dyn},m} &:= \alpha^2\max_{\Lambda_R, \Lambda_L}\bigg|\sum_{x\in \mc{X}}\!\braa{E_x\tn{2} \!\otimes\!\Lambda_L^*(\tilde E_x)}\big(P_\sigma^{(2)}(A\tn{2}\!\otimes \!\Lambda) P_\sigma^{(2)} \big)^{m\!-\!1}\kett{\rho\tn{2}\!\otimes\! \Lambda_R(\tilde\rho)}\bigg|,
\end{align}
of the shadow norm completes the argument.

\end{proof}

\section{Shadow norm bound for local Clifford UIRS}\label{appsec:local_var}

In this section, we consider the UIRS protocol with the local Clifford group $\mathbb{C}_1^{\times n}$. We will model the noisy implementation of any given Clifford by $\Lambda_L \omega(g) \Lambda_R$ and we will denote $\Lambda_R\Lambda_L=:\Lambda$ for brevity. In the main text we stated the following theorem:
\begin{theorem}[Restatement of Theorem 4 in the main text]
\label{thm:local_cliff_var_sm}
Consider the local Clifford UIRS protocol.
Let $\phi(g) = \Lambda_L\omega(g)\Lambda_R$ be a noisy implementation of the local Clifford group on $n$ qubits and $A = P_w A P_w $ be a probe super-operator with $|w|=k$ for $k$ a fixed integer. Also let $\tilde 0_n$ be a noisy implementation of the all-zero state and $\{\tilde x\}_x$ the noisy computational basis POVM). The shadow norm of the random variable $X_A(m)$ is upper bounded independently of the number of qubits $n$ and sequence length $m$. In particular, it holds that
\begin{equation}
\norm{A}_{{\rm dyn},m}\leq 2^k 3^{2k} \big(3^{-k}\tr(AA\ct)\big)^{m-1}.
\end{equation}
\end{theorem}
\begin{proof}

We begin from the general expression of the shadow norm, which can be written as
\begin{align}
\norm{A}_{{\rm dyn},m}&= \max_{\Lambda_R,\Lambda_L} 3^{2k}2^{2n} \sum_{x\in \{0,1\}^n}\!\!\!\!\!\braa{x^{\otimes2}\otimes \tilde x}(\mathbb{I}^{\otimes2}\otimes \Lambda_L) \left[ \avg_{g_{[1,n]} } \omega(g_{[1,n]})^{\otimes 3} (A^{\otimes2}\otimes \Lambda) \avg_{g'_{[1,n]} } \omega(g'_{[1,n]})^{\otimes 3} \right]^m \!\!\!\!(\mathbb{I}^{\otimes2}\otimes \Lambda_R)\kett{0_n^{\otimes2}\otimes \tilde{0}_n} \notag\\
&=\max_{\Lambda_R,\Lambda_L}
3^{2k}2^{2n} \avg_{g^{(1)}_{[1,n]},\ldots , g^{(m)}_{[1,n]} } \sum_{x\in \{0,1\}^n}\braa{x^{\otimes 2}}  \omega(g^{(m)}_{[1,n]})^{\otimes 2}\prod_{i =1}^{m-1}\!\big(A^{\otimes2}\omega(g^{(i)}_{[1,n]})^{\otimes 2}\big) \kett{0_n^{\otimes 2}}\notag\\
&\hspace{30em}\times \braa{\tilde{x}} \Lambda_L \omega(g^{(m)}_{[1,n]})\prod_{i=1}^{m-1}\!\big(\Lambda \omega(g^{(i)}_{[1,n]})\big) \Lambda_R\kett{\tilde{0}_n},\notag
\end{align}
where $\{\braa{\tilde x}\}_x$ and $\kett{\tilde{0}_n}$ are the noisy measurement POVM and the noisy initial state, respectively.

We now make use of the fact that $A$ is assumed to be supported on only a single irreducible representation denoted by $k\in \{0,1\}^n$ i.e., $A=P_{k} A P_{k}$ where $P_{k}$ is the projector onto that irreducible representation. Without loss of generality we will here set $k$ to be the all $1$ bit string on the first $k$ bits and $0$ on the remaining $n-k$ bits. The projector $P_w$ acts as $\kett{\tau_0}\braa{\tau_0}$ on these last $n-k$ qubits. Since
\begin{equation}
 \omega(g)\kett{\tau_0} = \kett{\tau_0},
\end{equation}
we see that in $\braa{x^{\otimes 2}}  \omega(g^{(m)}_{[1,n]})^{\otimes 2}\prod_{i =1}^{m-1}\big(A^{\otimes2}\omega(g^{(i)}_{[1,n]})^{\otimes 2}\big) \kett{0_n^{\otimes 2}}$ we can absorb the action of the local Clifford group on the last $n-k$ qubits. Hence, the local Cliffords $g^{(i)}_{[k+1,n]}$ only act by a single conjugation, i.e.,
\begin{align}
\norm{A}_{{\rm dyn},m}&= \max_{\Lambda_R,\Lambda_L}3^{2k}2^{2n} \avg_{g^{(1)}_{[1,n]},\ldots , g^{(m)}_{[1,n]} } \sum_{x\in \{0,1\}^n}
\braa{x^{\otimes 2}}  \omega(g^{(m)}_{[1,k]})^{\otimes 2}\prod_{i =1}^{m-1}\big(A^{\otimes2})\omega(g^{(i)}_{[1,k]})^{\otimes 2}\big) \kett{0_n^{\otimes 2}}\notag\\
&\hspace{10em}\times \braa{\Lambda_L^*(\tilde{x})}  \omega(g^{(m)}_{[1,k]})\otimes\omega(g^{(m)}_{[k+1,n]})\prod_{i=1}^{m-1}\big(\Lambda \omega(g^{(i)}_{[1,k]}) \omega(g^{(i)}_{[k+1,n]})\big) \kett{\Lambda_R(\tilde{0}_n)}\notag\\
&\hspace{12em} \times \big[\braakett{x_{[k+1,n]}}{{\tau_0}_{[k+1,n]}}\braakett{{\tau_0}_{[k+1,n]}}{{0_n}_{[k+1,n]}}\big]^2.
\end{align}
 Now we use the fact that
\begin{equation}
\braakett{x_{[k+1,n]}}{{\tau_0}_{[k+1,n]}} =2^{(k-n)/2}, \ \;\;\;\;\; x_{[k+1,n]}\in \{0,1\}^{n-k},
\end{equation}
and
\begin{equation}
\avg_{g_{[k+1,n]}}\omega(g_{[k+1,n]}) = \kett{{\tau_0}_{[k+1,n]}}\braa{{\tau_0}_{[k+1,n]}},
\end{equation}
we end up with
\begin{align}
\norm{A}_{{\rm dyn},m}&= \max_{\Lambda_R,\Lambda_L}3^{2k}2^{2k}  \avg_{g^{(1)}_{[1,k]},\ldots , g^{(m)}_{[1,k]} } \sum_{x_{[1,k]}\in \{0,1\}^k}
\braa{x^{\otimes 2}_{[1,k]}}  \omega(g^{(m)}_{[1,k]})^{\otimes 2}\prod_{i =1}^{m-1}\big(A^{\otimes2}\omega(g^{(i)}_{[1,k]})^{\otimes 2}\big) \kett{0_k^{\otimes 2}}\\
&\hspace{10em}\times\!\!\!\!\! \sum_{x_{[k+1,n]}\in \{0,1\}^{n-k}}\!\!\!\!\!\!\!2^{k-n}\braa{\tr_{k+1,n}(\Lambda_L^*(\tilde{x}))}  \omega(g^{(m)}_{[1,k]})\prod_{i=1}^{m-1}\big(\Lambda \omega(g^{(i)}_{[1,k]})\big) \kett{\tr_{k+1,n}(\Lambda_R(\tilde{0}_n))}.
\notag
\end{align}

Obviously $ \kett{\tilde{\rho}_k}:= \kett{\tr_{k+1,n}(\Lambda_R(\tilde{0}_n))}$ is a $k$-qubit state and, moreover,
we have
\begin{align}
\sum_{x_{[1,k]}\in\{0,1\}^k} \tilde{E}_{x_{[1,k]}} &:=\sum_{x_{[1,k]}\in\{0,1\}^k} \bigg(2^{k-n}\sum_{x_{[k+1,n]}\in\{0,1\}^{n-k} }\tr_{k+1,n}(\Lambda_L^*(\tilde{x})) \bigg)\\
& =  2^{k-n} \sum_{x\in \{0,1\}^n} \tr_{k+1,n}(\Lambda_L^*(\tilde x))\nonumber \\
&= 2^{k-n}\tr_{k+1,n}(\Lambda_L^*(\mathbb{I}))\nonumber \\
&= 2^{k-n}\tr_{k+1,n}(\mathbb{I})  = \mathbb{I}_{[1,k]}, \nonumber
\end{align}
which makes $\{\tilde{E}_{x_{[1,k]}}\}_{{x_{[1,k]}}}$ a $k$-qubit POVM. Hence, we have
\begin{align}
\norm{A}_{{\rm dyn},m}&= \max_{\Lambda_R,\Lambda_L}3^{2k}2^{2k}  \avg_{g^{(1)}_{[1,k]},\ldots , g^{(m)}_{[1,k]} } \sum_{x_{[1,k]}\in \{0,1\}^k}
\braa{x^{\otimes 2}_{[1,k]}}  \omega(g^{(m)}_{[1,k]})^{\otimes 2}\prod_{i =1}^{m-1}\big(A^{\otimes2}\omega(g^{(i)}_{[1,k]})^{\otimes 2}\big) \kett{0_k^{\otimes 2}}\notag\\
&\hspace{10em}\times \braa{\tilde{E}_{x_{[1,k]}}}  \omega(g^{(m)}_{[1,k]})\prod_{i=1}^{m-1}\big(\Lambda_k \omega(g^{(i)}_{[1,k]})\big) \kett{\tilde{\rho}_k} ,
\end{align}
where $\Lambda_k$ is the $k$-qubit marginal of $\Lambda$.
At this point, we see that the shadow norm is bounded independently of the number of qubits $n$ and only depends on the dimension of the irreducible representation which we assume $A$ to have support on (it is a function of $k$ only). As we can see, we are left with a third moment calculation over $\mathbb{C}_1^{\times k}$. However, for the sake of obtaining an upper bound we simply note that
\begin{equation}
\braa{\tilde{E}_{x_{[1,k]}}}  \omega(g^{(m)}_{[1,k]})\prod_{i=1}^{m-1}\big(\Lambda_k \omega(g^{(i)}_{[1,k]})\big) \kett{\tilde{\rho}_k} \in [0,1],
\end{equation}
since $\Lambda_k$ is a quantum channel.
We, therefore, can (using the invariance of the Haar measure) simplify the bound to
\begin{equation}
\norm{A}_{{\rm dyn},m}\leq 3^{2k}2^{3k} \avg_{g_{[1,k]}}\avg_{g'_{[1,k]}} \braa{0_{[1,k]}^{\otimes2}} \left[\omega(g_{[1,k]})^{\otimes 2}A^{\otimes2} \omega(g'_{[1,k]})^{\otimes 2}\right]^{m-1}\kett{0_{[1,k]}^{\otimes2}}.
\end{equation}
Again using the fact that $A$ is taken to only have overlap with a single irreducible
representation, \cref{second_moment_local_clifford} and the fact that $\braakett{0}{\tau_i} =1/\sqrt{2} $ if and only if $\tau_i=\tau_0$ or $\tau_Z$ (and zero otherwise) we have
\begin{align}
\norm{A}_{{\rm dyn},m}&\leq 3^{2k}2^{3k} \left[\frac{\tr(AA^{\dagger})}{3^{k}} \right]^{m-1} \braakett{0^{\otimes 2k}}{\tau_Z^{\otimes2k}}^2 \\
&=3^{2k}2^k \left[\frac{\tr(AA^{\dagger})}{3^{k}} \right]^{m-1},\nonumber
\end{align}
which is what we set out to prove.
\end{proof}

\section{Shadow norm bound for 
multi-qubit Clifford UIRS}\label{appsec:multi_var}

In this section, we will go into the details of the shadow norm bound calculations for the multi-qubit Clifford group UIRS protocol. Concretely, we prove the following theorem.

\begin{theorem}[Restatement of Theorem~3 in the main text]\label{thm:cliff_var_sm}
Consider the $n$-qubit Clifford UIRS protocol and let $A = P_{\rm ad} AP_{\rm ad} $ be a probe super-operator restricted to the traceless subspace. Also let $\tilde 0_n$ be a noisy implementation of the all-zero state and $\{\tilde x\}_{x\in \{0,1\}^n}$ the noisy computational basis POVM). The associated shadow norm is upper bounded as
\begin{align}
  \norm{A}_{{\rm dyn},m}&\leq  11\,u(A)\bigg( r(A)^{m-2} + \big[2(m-2)^2r(A)^{m-3} \big]   \max\big\{ 11u(A), (11\, u(A))^2\big\}  \! \bigg),
\end{align}
with
\begin{equation}
r(A) = u(A)(1+ 16 \, 2^{-n/3}),
\end{equation}
and where $u(A) = \tr(AA\ct)/(2^{2n}-1)$ is the unitarity of $A$.
\end{theorem}

Restricting to $m > 0$ and using that $r(A) \geq u(A)$ yields the simplified statement in the main text.
\begin{proof}
 Recall that the dynamic shadow norm for the  multi-qubit Clifford group is given by
 \begin{align}
\norm{A}_{{\rm dyn},m} = \max_{\Lambda_R,\Lambda_L} (2^n+1)^2 \!\!\!\!\!\sum_{x\in \{0,1\}^n}\!\!\!\!\! \tr\!\bigg[&  \kett{0_n\tn{2}\otimes\Lambda_R(\tilde{0}_n)}\!\braa{x\tn{2}\otimes \Lambda_L^*(\tilde x)}\avg_{g\in \gr} \sigma_{\rm ad} (g)^{\otimes 2}\otimes \omega(g)\bigg[\avg_{g\in \gr} (A\sigma_{\rm ad}(g))\tn{2} \!\otimes\! (\Lambda\omega(g)) \bigg]^{m\!-\!1} \bigg],
\end{align}
where $\Lambda:= \Lambda_R\Lambda_L$ is the noise in-between subsequent gates.
We can provide a concrete resolution for the third moment by noting that $\mathbb{C}_n$ is a $3$-design~\cite{PhysRevA.96.062336}, and, hence, its third moment follows that of the unitary group $U(2^n)$, which is fully determined (for $n\geq2$) by Schur-Weyl duality. In particular, we have
\begin{equation}
\avg_{c\in \mathbb{C}_n} \sigma_{\rm ad}(c)\tn{2} \otimes \omega(c) = P_{\rm ad}\tn{2}\otimes \mathbb{I}\bigg[\avg_{c\in \mathbb{C}_n} \omega(c)\tn{3}\bigg]P_{\rm ad}\tn{2}\otimes \mathbb{I} = \sum_{\pi,\pi'\in S_3} W_{\pi,\pi'}(P_{\rm ad}\tn{2}\otimes \mathbb{I})\kett{\pi}\braa{\pi'}(P_{\rm ad}\tn{2}\otimes \mathbb{I}),
\end{equation}
where the matrices $\pi$ permute copies of the base Hilbert space, i.e.,
\begin{equation}
\pi\ket{i_1,i_2,i_3} = \ket{i_{\pi(1)},i_{\pi(2)},i_{\pi(3)} }.
\end{equation}
The Weingarten matrix
\begin{equation}\label{eq:weingarten}
W = \frac{1}{2^n(2^{2n}-1)(2^{2n}-4)}\begin{pmatrix}
2^{2n}-2 & -2^n     & -2^n      & -2^n     & 2         & 2       \\
-2^n     & 2^{2n}-2 & 2         & 2        & -2^n      & -2^n    \\
-2^n     & 2        & 2^{2n}-2  & 2        & -2^n      & -2^n    \\
-2^n     & 2        & 2         & 2^{2n}-2 & -2^n      & -2^n    \\
2        & -2^n     & -2^n      & -2^{n}   & 2^{2n}-2  & 2       \\
2        & -2^n     & -2^n      & -2^n     & 2         & 2^{2n}-2\\
\end{pmatrix}
\end{equation}
 can be explicitly written down in the basis $\{e, (12), (23), (13), (123), (132)\}$.
Now defining the matrices
\begin{equation}
\Omega_{\pi',\pi} = \braa{\pi'} A\tn{2}\otimes \Lambda \kett{\pi},
\end{equation}
and
\begin{equation}
[{\Theta_{x}}]_{\pi,\pi'} = (2^n+1)^2\braakett{\pi}{0_n\tn{2}\otimes\Lambda_R(\tilde{0}_n)}\!\braakett{x\tn{2}\otimes \Lambda_L^*(\tilde x)}{\pi'},
\end{equation}
we get
\begin{equation}
\norm{A}_{{\rm dyn},m} = \max_{\Lambda_R,\Lambda_L} \sum_{x\in\{0,1\}^n}\tr(\Theta_x (W\Omega)^{m-1}W).
\end{equation}
We begin by analyzing the matrix $\Omega$.
Note that $A(\mathbb{I}) = A\ct(\mathbb{I}) = 0$, since $A$ is supported only on the space of traceless matrices by construction. This means that $\Omega_{\pi,\pi'}$ is zero unless $\pi,\pi'\in \{(12),(123),(132)\}$. Thus we can write $\Omega = P\ct \hat{\Omega}P$ with $P$ the restriction from $\text{Span}\{e, (12), (23), (13), (123), (132)\}$ to $\text{Span}\{(12),(123),(132)\}$ and
\begin{equation}\label{eq:omega_restr}
\hat{\Omega} = \begin{pmatrix} 2^n \tr(A A\ct) & \tr(AA\ct) & \tr(AA\ct)\\
\tr(J_u(A)^2 \mathbb{I}\otimes \Lambda(\mathbb{I})) & \tr(J_u(A)^2 J_u(\Lambda)) & \tr(J_u(AT)^2 J_u(\Lambda T))\notag \\
\tr(J_u(A)^2 \mathbb{I}\otimes \Lambda(\mathbb{I})) & \tr(J_u(TA)^2 J_u(T\Lambda)) & \tr(J_u(A)^2 J_u(\Lambda))\notag\end{pmatrix},
\end{equation}
where we have used $\Lambda\ct(\mathbb{I}) = \mathbb{I}$, $J_u$ denotes the unnormalized Choi-isomorphism ($J_u(A) = \mathbb{I}\otimes A(\ket{v(\mathbb{I})}\bra{v(\mathbb{I})}$ with $v$ the column-stacking vectorization map) and $T$ is the (non-CP) transposition map. This can be derived directly from the definition of $\Omega$ and some diagram chasing. Before we continue, we establish some facts about the entries of $\hat{\Omega}$. We begin by noting that for the off-diagonal term $\hat{\Omega}_{(123),(132)}$, we have
\begin{equation}\label{eq:tn}
 \tr(J_u(AT)^2 J_u(\Lambda T)) =  \tr(J_u(TA)^2 J_u(T\Lambda )) =  \tr(J_u(A\ct T)^2 J_u(\Lambda\ct T)).
\end{equation}
This follows from three facts: (1) $A\tn{2}$ commutes with the super-operator $L_{(12)}$ defined as left-multiplication with the matrix $(12)$, (2) $L_{(12)}((123)) = (132)$ and (3) the trace is invariant under Hermitian conjugation. One can also see this graphically through the following series of tensor manipulations: as
\begin{equation*}
\includegraphics[scale=0.23]{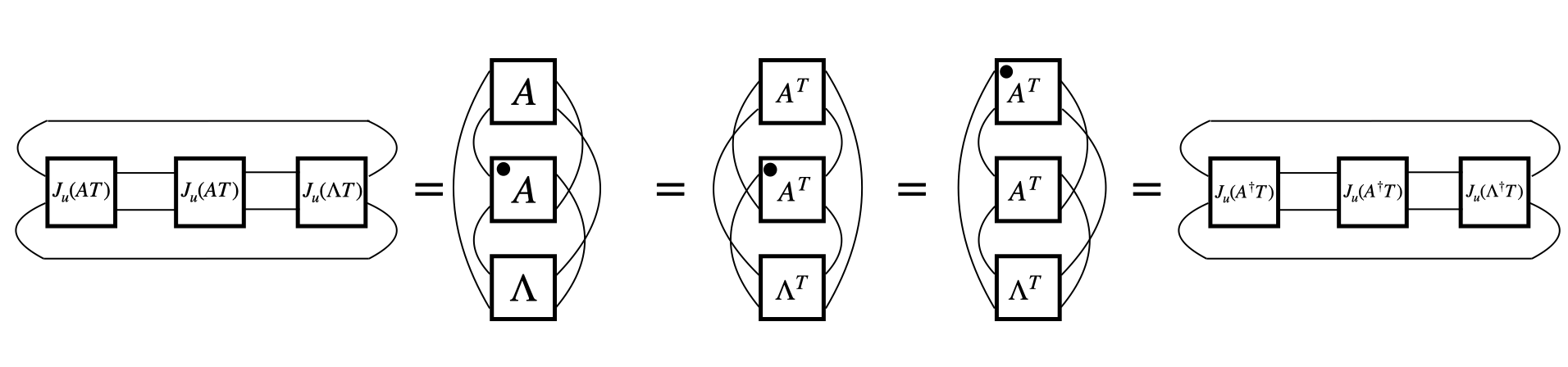}
\end{equation*}
where we have used that the Liouville representation of a dual super-operator $\Lambda\ct$ is given by the transpose of the Liouville representation of $\Lambda$. For clarity we marked for clarity one of the copies of $A$ with a dot.
The rightmost side of \cref{eq:tn} is amenable to a bound of the form
\begin{align}\label{eq:off-diag bound}
|\tr(J_u(A\ct T)^2 J(\Lambda\ct T))|&\leq \tr(J_u(A\ct T)^2) \norm{J_u(\Lambda\ct T)}_{\infty}\\
&= \tr(AA\ct) \norm{J_u(\Lambda\ct T)}_{\infty}\notag \\
&= \tr(AA\ct) \max_{\norm{Q}_1=1} |\tr( Q \,\mathbb{I} \otimes\Lambda\ct(\ket{v(\mathbb{I})}\bra{v(\mathbb{I})}^{T_B}))|\notag \\
&=\tr(AA\ct)\max_{\norm{Q}_1=1} |\tr( \mathbb{I} \otimes\Lambda(Q)(\ket{v(\mathbb{I})}\bra{v(\mathbb{I})}^{T_B}))| \notag \\
&\leq \tr(AA\ct) \norm{(\ket{v(\mathbb{I})}\bra{v(\mathbb{I})}^{T_B})}_{\infty} \notag \\
&\leq \tr(AA\ct),\notag
\end{align}
where we have used H{\"o}lders inequality twice, the fact that $A$ is Hermiticity preserving, and the fact that $\Lambda$ is $1\to 1$-norm contractive (by grace of being CPTP). Finally the Schatten $\infty$-norm of $(\ket{v(\mathbb{I})}\bra{v(\mathbb{I})}^{T_B})$ is easily seen to be one. Furthermore, for the diagonal elements $\hat{\Omega}_{(123),(123)}=\hat{\Omega}_{(132),(132)}$, we have
\begin{align}\label{eq:diag bound}
\tr(J_u(A)^2 J_u(\Lambda)) &\leq \tr(AA\ct)\norm{J_u(\Lambda)}_{\infty}\leq 2^n\tr(AA\ct)
\end{align}
again using H{\"o}lder's inequality and $\norm{J_u(\Lambda)}_{\infty}\leq 2^n $ (note again that $J_u$ is unnormalized). Finally, we consider the other off-diagonal term $\hat{\Omega}_{(123),(12)}=\hat{\Omega}_{(132),(12)} $. If $\Lambda$ is unital (and thus $\Lambda(\mathbb{I}) = \mathbb{I}$ ), we see that $\hat{\Omega}_{(123),(12)} = \tr(AA\ct)$, and in general we have
\begin{equation}
\tr(J_u(A)^2 \mathbb{I}\otimes \Lambda(\mathbb{I}))\leq \tr(AA\ct) \norm{\mathbb{I}\otimes \Lambda(\mathbb{I})}_\infty\leq 2^n \tr(AA\ct).
\end{equation}
Note that if $\Lambda$ is unital the matrix $\hat{\Omega}$ is symmetric, however, for general channels this is not the case. 

With these expressions we can move on to calculate the dynamic shadow norm
\begin{equation}
\norm{A}_{{\rm dyn},m} = \max_{\Lambda_R,\Lambda_L} \sum_{x\in \{0,1\}^n}\tr(\Theta_x (W\Omega)^{m-1}W) \leq \norm{\hat{\Omega}PW\big (\sum_{x\in \{0,1\}^n}\Theta_x\big)W P\ct}_{\infty} \norm{\hat{\Omega}\hat{W}^{m-2} }_{\infty},
\end{equation}
where $\hat{W} =PWP\ct$ and the norm is the standard operator (or Schatten $\infty-$) norm.
To bound the norm-of-power factor (the `dynamic' part of the shadow norm) in the above equation we need two useful lemmas about small non-normal matrices.

\begin{lemma}\label{lem:spectral difference}
Let $A,B$ be complex $d\times d$ matrices with spectral radius $s(A),s(B)$. We have that
\begin{equation}
|s(A) - s(B)| \leq (\norm{A}_{\infty} + \norm{B}_{\infty})^{1 - 1/d} \norm{A-B}_{\infty}^{1/d}.
\end{equation}
\end{lemma}
\begin{lemma}\label{lem:spectral power}
Let $A$ be a complex $d\times d$ matrix with spectral radius $s(A)$. We have that
\begin{equation}
\norm{A^m}_{\infty} \leq s(A)^m + (d-1) m^{d-1}\max \big\{ s(A)^{m-d+1}, s(A)^{m-1} \big\}
\max\big\{ 2\norm{A}_{\infty}, (2\norm{A}_{\infty})^{d-1}\big \}.
\end{equation}
\end{lemma}
\Cref{lem:spectral difference} is a trivial consequence of ref.~\cite[Thm VIII.1.1]{bhatia2013matrix}, and we will prove \cref{lem:spectral power} in at the end of this section. Next, we define matrices
\begin{equation}
\hat{\Omega}^0 = \begin{pmatrix} 2^n \tr(A A\ct) & 0 & 0\\
\tr(J_{u}(A)^2 \mathbb{I}\otimes \Lambda(\mathbb{I})) & \tr(J_{u}(A)^2 J_{u}(\Lambda)) & \tr(J_{u}(AT)^2 J_{u}(\Lambda T))\\
\tr(J_{u}(A)^2 \mathbb{I}\otimes \Lambda(\mathbb{I})) & \tr(J_{u}(TA)^2 J_{u}(T\Lambda)) & \tr(J_{u}(A)^2 J_{u}(\Lambda))\end{pmatrix}, \;\;\;\;\;\hat{D} = \frac{2^{2n}-2}{2^n(2^{2n}-1)(2^{2n}-4)}\mathbb{I}.
\end{equation}
We have
\begin{equation}\label{eq:spec_diff}
\hat{\Omega}^0\hat{D} =  \hat{\Omega}\hat{W} + (\hat{\Omega}^0-\hat{\Omega})\hat{D} + \hat{\Omega}(\hat{D} - \hat{W}).
\end{equation}
We can directly calculate the spectrum of $\hat{\Omega}^0\hat{D}$ as being
\begin{align}
\mathrm{Spec}(\hat{\Omega}^0\hat{D}) = &\bigg\{\frac{2^n (2^{2n}-2)\tr(AA\ct)}{2^n(2^{2n}-1)(2^{2n}-4)},\frac{(2^{2n}-2)}{2^n(2^{2n}-1)(2^{2n}-4)}\big(\tr(J_{u}(A)^2 J_{u}(\Lambda)) + \tr(J_{u}(TA)^2 J_{u}(T\Lambda))\big)\notag\\&\hspace{12em} , \frac{(2^{2n}-2)}{2^n(2^{2n}-1)(2^{2n}-4)}\big(\tr(J_{u}(A)^2 J_{u}(\Lambda)) - \tr(J_{u}(TA)^2 J_{u}(T\Lambda))\big)  \bigg\}.
\end{align}
Using the definition of unitarity we see that the first eigenvalue is upper bounded by
\begin{equation}
\frac{2^n (2^{2n}-2)\tr(AA\ct)}{2^n(2^{2n}-1)(2^{2n}-4)} =\frac{2^n (2^{2n}-2)(2^{2n}-1)}{2^n(2^{2n}-1)(2^{2n}-4)}u(A) \leq (1+ 3 \,2^{-2n})u(A),
\end{equation}
for $n\geq2$.
Note that the constant $3$ in the last inequality is somewhat loose.
Furthermore, using \cref{eq:diag bound,eq:off-diag bound} we can also provide bounds for the other two eigenvalues, to get
\begin{align}
\frac{(2^{2n}-2)}{2^n(2^{2n}-1)(2^{2n}-4)}\big(\tr(J_{u}(A)^2 J_{u}(\Lambda)) \pm \tr(J_{u}(TA)^2 J_{u}(T\Lambda))\big)
&\leq \frac{(2^{2n}-2)(2^n+1)}{2^n(2^{2n}-1)(2^{2n}-4)}\tr(AA\ct) \\
&=\frac{(2^{2n}-2)(2^n+1)}{2^n(2^{2n}-4)}u(A)\notag \\
&\leq (1+ 4\, 2^{-n})u(A) ,\notag
\end{align}
where again the constant $4$ is an overestimate.
Hence, the spectral radius of $\hat{\Omega}^0$ is bounded by $u(A)$ up to a small multiplicative correction, i.e.
\begin{equation}\label{eq:spec_rad}
s(\hat{\Omega}^0) \leq (1+ 4\, 2^{-n})u(A).
\end{equation}
The plan is now to leverage \cref{lem:spectral difference} to bound the spectral radius of $\hat{\Omega}\hat{W}$. To do this we first need to bound the norm (from \cref{eq:spec_diff})
\begin{equation}
\norm{\hat{\Omega}^0\hat{D} - \hat{\Omega}\hat{W}}_{\infty} \leq \norm{(\hat{\Omega}-\hat{\Omega}^0)}_{\infty}\norm{\hat{D}}_{\infty} + \norm{\hat{\Omega}}_{\infty}\norm{(\hat{D} - \hat{W})}_{\infty}.
\end{equation}
We can bound the terms on the RHS in a straightforward manner (using the definition \cref{eq:weingarten}),
\begin{equation}
\norm{(\hat{D} - \hat{W})}_{\infty} \leq \norm{(\hat{D} - \hat{W})}_{HS} = \frac{\sqrt{4\, 2^{2n} + 8}}{2^n(2^{2n}-1)(2^{2n}-4)},
\end{equation}
and (using the definition of $\hat{\Omega}^0$)
\begin{equation}
\norm{(\hat{\Omega}-\hat{\Omega}^0)}_{\infty} \leq (2^{2n}-1)u(A),
\end{equation}
and finally
\begin{align}
\norm{\hat{\Omega}}_{\infty} &\leq \norm{\hat{\Omega}}_{HS} \leq3\, 2^{n}(2^{2n}-1)u(A),\\
\norm{\hat{D}}_{\infty} &= \frac{2^{2n}-2}{2^n(2^{2n}-1)(2^{2n}-4)}.
\end{align}
Plugging this all back in we get
\begin{equation}
\norm{\hat{\Omega}^0\hat{D} - \hat{\Omega}\hat{W}}_{\infty}  \leq u(A) \frac{3\, 2^{n}(2^{2n}-1)\sqrt{4\, 2^{2n} + 8} + (2^{2n}-2)(2^{2n}-1)}{2^n(2^{2n}-1)(2^{2n}-4)}\leq 3\,2^{-n}u(A).
\end{equation}
Moreover, we have that
\begin{equation}
\norm{\hat{\Omega}^0\hat{D}}_{\infty}\leq \frac{9\, u(A) 2^n(2^{2n}-1)(2^n-2)}{2^n(2^{2n}-1)(2^{2n}-4)} \leq 11\,u(A),
\end{equation}
for $n\geq 3$,
and by the same argument
\begin{equation}
\norm{\hat{\Omega}\hat{W}}_{\infty} \leq 11\,u(A).
\end{equation}
Hence, through \cref{lem:spectral difference}, the spectral radius difference can be bounded as
\begin{equation}
|s(\hat{\Omega}^0\hat{D}) - s(\hat{\Omega}\hat{W})| \leq \big(22 u(A)\big)^{2/3}\big(3\, 2^{-n}u(A)\big)^{1/3} \leq 12\, 2^{-n/3}u(A),
\end{equation}
and the spectral radius of $\hat{\Omega}\hat{W}$ is thus bounded by (using the above equation and \cref{eq:spec_rad})
\begin{equation}
s(\hat{\Omega}\hat{W}) \leq (1 + 4\, 2^{-n} + 12\, 2^{-n/3})u(A).
\end{equation}
We can plug this into \cref{lem:spectral power} to obtain
\begin{equation}
\norm{\hat{\Omega}\hat{W}^{m-2} } \leq  (u(A)(1+16\, 2^{-n/3}))^{m-2} + \big[2(m-2)^2(u(A)(1+16\, 2^{-n/3}))^{m-3} \big]   \max\big\{ 11u(A), (11\, u(A))^2\big\},
\end{equation}
which takes care of the `dynamic' part of the dynamic shadow norm. To bound the SPAM contribution we can observe
\begin{equation}
\norm{\hat{\Omega}P W \big(\!\!\!\sum_{x\in \{0,1\}^n}\!\!\!\Theta_x\big)W P\ct}_{\infty}\leq \norm{\hat{\Omega}}_{HS} \norm{W}_{HS}^2 \norm{\sum_{x\in \{0,1\}^n}\!\!\!\Theta_x}_{HS}.
\end{equation}
Of these only $\|{\sum_{x\in \{0,1\}^n}\Theta_x}\|_{HS}$ has not been considered. From a straightforward calculation (using that $\Lambda_R(\tilde{0}_n)$ is a state and $\{\Lambda_L^*(\tilde x)\}_x$ a POVM) one sees that $\sum_{x\in \{0,1\}^n}\Theta_x = 2^n(2^n+1)^2 v_1 v_2^T$ with
\begin{align}
v_1 &:= \Big(1,1,\avg_{x\in \{0,1\}^n}\braakett{x}{\Lambda_L^*(\tilde x)}, \avg_{x\in \{0,1\}^n}\braakett{x}{\Lambda_L^*(\tilde x)},\avg_{x\in \{0,1\}^n}\braakett{x}{\Lambda_L^*(\tilde x)},\avg_{x\in \{0,1\}^n}\braakett{x}{\Lambda_L^*(\tilde x)}\Big),\\
v_2 &:= \Big(1,1,\braakett{0_n}{\Lambda_R(\tilde{0}_n)}, \braakett{0_n}{\Lambda_R(\tilde{0}_n)},\braakett{0_n}{\Lambda_R(\tilde{0}_n)},\braakett{0_n}{\Lambda_R(\tilde{0}_n)}\Big).
\end{align}
We can use this to upper bound the SPAM factor as
\begin{align}
\norm{\hat{\Omega}PW\big (\!\!\!\sum_{x\in \{0,1\}^n}\!\!\!\Theta_x\big)W P\ct}_{\infty} &\leq \norm{\hat{\Omega}PW\big (\sum_{x\in \{0,1\}^n}\Theta_x\big)W P\ct}_{HS}\\
\notag
 &= 2^{n}(2^n +1)^2 \big[v_1 W\ct P\ct \hat{\Omega}\ct\hat{\Omega}PW v_1\ct \; v_2 W\ct P\ct W v_2\ct\big]^{1/2}\\
 &\leq 2^{n}(2^n +1)^2\bigg[ \norm{\hat{\Omega}\hat{\Omega}\ct}_{\infty} v_1 W\ct P\ct P Wv_1\ct \; v_2 W\ct P\ct W v_2\ct \bigg ]^{1/2}\notag \\
 \notag
& \leq \frac{2^{n}(2^n +1)^2 \norm{\hat{\Omega}}_{HS}}{\big(2^n (2^{2n}-4)(2^{2n}-2)\big)^2}
\notag
 \bigg[\big[(2^{2n} -2^{2n}-2) + \avg_{x\in \{0,1\}^n} \braakett{x}{\tilde x} (4-2\, 2^{n})\big]^2\notag\\
 &\hspace{12em}+ 2\big[(2-2^n) + \avg_{x\in \{0,1\}^n} \braakett{x}{\tilde x} (2^{2n} -2^n)\big]^2\bigg]^{1/2} \notag\\
 &\hspace{12em} \times \bigg[ \big[(2^{2n} -2^{2n}-2) + \braakett{0_n}{\tilde 0_n} (4-2\, 2^{n})\big]^2 \notag\\
 &\hspace{13em}+ 2\big[(2-2^n) + \braakett{0_n}{\tilde 0_n} (2^{2n} -2^n)\big]^2\bigg]^{1/2}.\notag
\end{align}
Gathering terms, throwing away some negative ones, remembering our bound for $\norm{\hat{\Omega}}_{HS}$, and using that
\begin{align}
\avg_{x\in \{0,1\}^n}\braakett{x}{\Lambda_L^*(\tilde x)}\leq 1,
\end{align}
and $\braakett{0_n}{\Lambda_R(\tilde{0}_n)}\leq 1$ by construction we can bound this further by
\begin{align}
\norm{\hat{\Omega}PW\big (\!\!\!\sum_{x\in \{0,1\}^n}\!\!\!\Theta_x\big)W P\ct}_{\infty} \leq \frac{3\, 2^{2n}(2^n +1)^2 (2^{2n}-1) }{\big(2^n (2^{2n}-4)(2^{2n}-2)\big)^2}\bigg[(2^{2n} - 2^n -2)^2 + 2(2-2^n)^2 + (2^{2n}- 2^n)^2 + (4-2\, 2^n)^2 \bigg].
\end{align}
We would like to stress that from this expression we can already see that the SPAM norm contribution is asymptotically independent of $n$. By basic numerics, we can obtain
\begin{equation}
\norm{\hat{\Omega}PW\big (\!\!\!\sum_{x\in \{0,1\}^n}\!\!\!\Theta_x\big)W P\ct}_{\infty} \leq 11\, u(A) .
\end{equation}
We again note that this constant is sub-optimal (especially for large $n$). Putting all of this together we obtain the stated bound.
\end{proof}

When the probe super-operator is a unitary restricted to the traceless subspace ($A = P_{\rm ad} UP_{\rm ad}$ for some unitary $U$) then we can obtain a substantially improved bound, which hinges critically on the fact that $U$ is a quantum channel.

\begin{theorem}[Restatement of Theorem 2 in the main text]
\label{thm:var_unitary_sm}
Consider the $n$-qubit Clifford UIRS protocol and let $A = P_{\rm ad} UP_{\rm ad} $ be a probe super-operator with $U$ a unitary. Also let $\tilde 0_n$ be a noisy implementation of the all-zero state and $\{\tilde x\}_x$ the noisy computational basis POVM). The dynamic shadow norm is bounded as
\begin{equation}
\norm{A}_{{\rm dyn},m} \leq   10.
\end{equation}
\end{theorem}
\begin{proof}
We begin from the definition of the dynamic shadow norm, given by
\begin{equation}
\norm{A}_{{\rm dyn},m} = \max_{\Lambda_R,\Lambda_L} (2^n+1)^2 2^n \sum_{x\in \{0,1\}^n} \braa{x\tn{2} \otimes \tilde x} \bigg[\avg_{g\in \mathbb{C}_n}\phi(g)\tn{3} (P_{\rm ad} UP_{\rm ad})\tn{2}\otimes \Lambda \avg_{g\in \mathbb{C}_n}\phi(g)\tn{3}\bigg]^m \kett{0_n\tn{2}\otimes \tilde 0_n}.
\end{equation}
Here, again we write $\Lambda := \Lambda_R\Lambda_L$.
Note also that $P_{\rm ad}^2 = P_{\rm ad}$ and that $P_{\rm ad}$ commutes with both $U$ and $\phi(g)$. Hence, we can write
\begin{equation}
\norm{A}_{{\rm dyn},m} = \max_{\Lambda_R,\Lambda_L} (2^n+1)^2 2^n \avg_{x\in \{0,1\}^n} \braa{x\tn{2} \otimes \Lambda_L^*(\tilde x)}(P_{\rm ad}\tn{2}\otimes \mathbb{I}) \bigg[\avg_{g\in \mathbb{C}_n}\phi(g)\tn{3} (U\tn{2}\otimes \Lambda)\! \avg_{g\in \mathbb{C}_n}\phi(g)\tn{3}\bigg]^m \kett{0_n\tn{2}\otimes \Lambda_R(\tilde 0_n)}.
\end{equation}
Furthermore $\smallavg_{g\in \mathbb{C}_n}\phi(g)\tn{3}$ is a projector, and thus
\begin{align}
\norm{A}_{{\rm dyn},m} &= \max_{\Lambda_R,\Lambda_L} (2^n+1)^2 2^n \!\!\!\!\avg_{x\in \{0,1\}^n}\!\!\!\! \braa{x\tn{2} \otimes \Lambda_L^*(\tilde x)}(P_{\rm ad}\tn{2}\otimes \mathbb{I})\sum_{\pi,\pi'}W_{\pi,\pi'}\kett{\pi}\braa{\pi'}\notag\\
&\hspace{15em}\times \bigg[\avg_{g\in \mathbb{C}_n}\phi(g)\tn{3} (U\tn{2}\otimes \Lambda)\! \avg_{g\in \mathbb{C}_n}\phi(g)\tn{3}\bigg]^m\kett{0_n\tn{2}\otimes\Lambda_R(\tilde 0_n)},
\end{align}
where we have used the resolution of the third-moment projector as in \cref{eq:weingarten}. Now defining the vectors
\begin{align}
\hat{v}_{\pi} &:= \braa{x\tn{2} \otimes \Lambda_L^*(\tilde x)}(P_{\rm ad}\tn{2}\otimes \mathbb{I})\kett{\pi},\\
\hat{w}_{\pi'}(m) &:= \braa{\pi'} \bigg[\avg_{g\in \mathbb{C}_n}\phi(g)\tn{3} (U\tn{2}\otimes \Lambda)\! \avg_{g\in \mathbb{C}_n}\phi(g)\tn{3}\bigg]^m \kett{0_n\tn{2}\otimes \Lambda_R(\tilde 0_n)},
\end{align}
we can express the shadow norm as
\begin{equation}
\norm{A}_{{\rm dyn},m} = \max_{\Lambda_R,\Lambda_L} (2^n+1)^2 2^n v W w^T(m) .
\end{equation}
By direct calculation analogous to the calculations done in the proof of \cref{thm:cliff_var_sm},
this becomes
\begin{align}
\norm{A}_{{\rm dyn},m} &= \max_{\Lambda_R,\Lambda_L}\frac{(2^n+1)^2 2^n}{2^n (2^{2n}-4)(2^{2n}-1)}\bigg[
\big((2^{2n}-2) - 2\avg_{x\in\{0,1\}^n}\braakett{x}{\Lambda_L^*(\tilde x)} 2^n\big)\hat{w}_{(12)}\notag\\
&\hspace{5em} +
\big(2^{2n}\avg_{x\in\{0,1\}^n}\braakett{x}{\Lambda_L^*(\tilde x)} -2^n\big)\hat{w}_{(123)} +
\big(2^{2n}\avg_{x\in\{0,1\}^n}\braakett{x}{\Lambda_L^*(\tilde x)} -2^n\big)\hat{w}_{(132)}
\bigg].
\end{align}
Finally, we use that $\norm{\pi'}_{\infty}=1$ and that $\kett{0_n\tn{2}\otimes \Lambda_R(\tilde 0_n)}$ is a quantum state to see that
\begin{align}
|\hat{w}_{\pi'}(m)| &= |\braa{\pi'} \bigg[\avg_{g\in \mathbb{C}_n}\phi(g)\tn{3} (U\tn{2}\otimes \Lambda)\! \avg_{g\in \mathbb{C}_n}\phi(g)\tn{3}\bigg]^m \kett{0_n \tn{2}\otimes \Lambda_R(\tilde 0_n)}|\\
\notag
&\leq \norm{\bigg[\avg_{g\in \mathbb{C}_n}\phi(g)\tn{3} (U\tn{2}\otimes \Lambda)\! \avg_{g\in \mathbb{C}_n}\phi(g)\tn{3}\bigg]^m}_{1\to1}\\
\notag
&\leq 1,
\notag
\end{align}
since $U\tn{2}\otimes\Lambda $ and $\smallavg_{g\in \mathbb{C}_n}\phi(g)\tn{3}$ are quantum channels.
Together with the fact that $|\smallavg_{x\in \{0,1\}^n} \braakett{x}{\Lambda_L^*(\tilde x)}|\leq 1$, we get
\begin{align}
\norm{A}_{{\rm dyn},m} \leq\frac{(2^n+1)^2 2^n}{2^n (2^{2n}-4)(2^{2n}-1)}\bigg[
(2^{2n}-2) + 2\,2^n\big) +
\big(2^{2n} +2^n\big) +
\big(2^{2n}+2^n\big)
\bigg]\leq 10
\end{align}
for $n\geq 2$.
\end{proof}

Finally, we provide a proof of \cref{lem:spectral power}, adapted from a very similar statement in ref.~\cite[Lemma 8.5]{wolf2012quantum}.

% \begin{lemma}[Restatement of \cref{lem:spectral power}]
% Let $A$ be a complex $d\times d$ matrix with spectral radius $s(A)$. We have that
% \begin{equation}
% \norm{A^m}_{\infty} \leq s(A)^m + (d-1) m^{d-1} \max \big\{ s(A)^{m-d+1}, s(A)^{m-1} \big\}
% \max\big\{ 2\norm{A}_{\infty}, (2\norm{A}_{\infty})^{d-1}\big \}.
% \end{equation}
% \end{lemma}
\begin{proof}[Proof of \cref{lem:spectral power}]
We begin by bringing $A$ into Schur normal form, i.e.,
\begin{equation}
A = U(D+ N)U\ct
\end{equation}
where $D$ is diagonal (with the eigenvalues of $A$ on the diagonal), $N$ is strictly upper triangular and $U$ is unitary.
Now consider the expansion of $(D+N)^m$. Since $N$ is strictly upper triangular, any term with more than $d-1$ factors of $N$ must vanish. Hence, we have
\begin{align}
\norm{A^m} &= \norm{U(D+N)^mU\ct}\\
\notag
& \leq  \norm{\Lambda}^m_{\infty} + \sum_{i=1}^{\min\{d-1, m\}} \binom{m}{k}\norm{D}_\infty^{m-i}\norm{N}_\infty^{i}\\
\notag
&\leq  \norm{\Lambda}^m_{\infty} + (d-1) m^{d-1}  \max \big\{ s(A)^{m-d+1}, s(A)^{m-1} \big\}
\max\big\{ \norm{N}_{\infty}, (\norm{N}_{\infty})^{d-1}\big \},
\notag
\end{align}
where we have used that $\binom{m}{k}\leq m^{d-1}$, H{\"o}lder's inequality, and the monotonicity of the exponential function. Now note that
\begin{equation}
\norm{N}_{\infty} \leq \norm{D}_\infty + \norm{A}_{\infty}\leq 2 \norm{A}_{\infty}.
\end{equation}
Combining this with $\norm{D}_{\infty} = s(D) = s(A)$ we obtain the lemma statement.
\end{proof}

\section{Pauli-noise estimation}\label{appsec:pauli_var}

Here we analyze the Pauli-noise estimation scheme outlined in the main text.
Recall that we consider sequences of the form $\mathbf g = (c^{-1}, p_m, \ldots, p_1, c)$, where $p_1, \ldots, p_m$ are i.i.d.\ randomly drawn elements from the Pauli group $\mathbb P_n$ and $c$ is randomly drawn multi-qubit Clifford $\mathbb C_n$.
For $\tau$ a traceless Hilbert-Schmidt normalized multi-qubit Pauli operator and sequence $\mathbf g$ we define the filter-function as
\begin{equation}
f_\tau(x,\vec{g}) \coloneqq \alpha \braa{x} \omega(c)\omega(p_m) A_{\tau}\ldots A_\tau\omega(p_1) \omega(c)\kett{0_n},
\end{equation}
with $A_\tau = \kett{\tau}\!\braa{\tau}$ and $\alpha = 2^n(1+2^n)$.
Without SPAM, the expected value of the single shot estimator, is given by
\begin{align}
k_\tau(m) &= \alpha\avg_{c\in \mathbb{C}_n}\sum_{x\in \{0,1\}^n}\!\!\!\!\braa{x\tn{2}} \omega(c^{-1})\tn{2}\kett{\tau\tn{2}}\braa{\tau\tn{2}}\omega(c\tn{2})\kett{0_n\tn{2}}\, \Lambda_{\tau,\tau}^{m-1} \\
&= \alpha \avg_{c\in \mathbb{C}_n} \sum_{x\in \{0,1\}^n} \sandwich x {\omega(c)(\tau)} x^2 \sandwich 0 {\omega(c)(\tau)} 0^2 \, \Lambda_{\tau,\tau}^{m-1}  \, .\notag
\end{align}
with  $\Lambda_{\tau,\tau} = \braa\tau\Lambda\kett\tau$.
Since the Clifford group acts transitively on the traceless Pauli operators $\mathbb P^\ast_n$, we can rewrite
\begin{equation}
  k_\tau(m) = \alpha \avg_{\tau' \in 2^{n/2} \mathbb P^\ast_n} \sum_{x \in \{0,1\}^n} \sandwich x {\tau'} x^{2} \sandwich 0 {\tau'} 0^{2}\, \Lambda_{\tau,\tau}^{m-1}.
\end{equation}
Out of the $2^{2n} - 1$ traceless Pauli-operators only $2^n -1$ have non-vanishing diagonal entries (those consisting only out of local $\mathbb I$ and $Z$). The non-vanishing diagonal entries are all identical to $2^{n/2}$.
Thus, using the definition of $\alpha$ we have
\begin{equation}
  k_\tau(m) = \alpha \frac{2^n(2^n - 1)}{2^{2n}(2^{2n} -1)} \Lambda_{\tau,\tau}^{m-1} = \Lambda_{\tau,\tau}^{m-1}\,.
\end{equation}

It remains to calculate the variance associated with estimating $k_\tau(m)$, as given in (19) in the main text. We intend to prove that $k_\tau(m)$
has variance bound in $O(2^n)$.
To do this first, note that
\begin{equation}
f_\tau(x,\vec{g})^2 = \braa{x\tn{2}} \omega(c^{-1})\tn{2} \kett{\tau\tn{2}}\braa{\tau\tn{2}}\omega(c)\tn{2}\kett{0_n\tn{2}}
\end{equation}
since $P\tn{2}\tau\tn{2}{P\ct}\tn{2} = \tau\tn{2}$. Hence, the variance associated to the
estimation can be upper bounded by
\begin{align}
\mathbb{V}_\tau(m)&\leq
2^{2n}(2^n+1)^2\!\!\!\!\sum_{x\in\{0,1\}^n} \avg_{c\in \mathbb{C}_n}\braa{x\tn{2}} \omega(c^{-1})\tn{2} \kett{\tau\tn{2}}\notag\\
&\hspace{6em}\times \braa{\tau\tn{2}}\omega(c)\tn{2}\kett{0_n\tn{2}} \braa{x} \omega(c^{-1})\kett{\tau_0}\notag\\
&\hspace{11em}\times\braa{\tau_0}\omega(c)\kett{0_n},
\end{align}
where we have used that $\mathbb E_{p\in \mathbb{P}}\ \omega(p) = \kett{\tau_0}\braa{\tau_0}$ and the trace preservation of $\Lambda$.
Noting again that $\omega(c)$ acts trivially on $\tau_0$ and transversally on the traceless Pauli operators, we see that
\begin{equation}
\!\!\!\mathbb{V}_\tau(m)\leq \frac{2^{3n}(2^n\!+\!1)^2}{2^{2n}-1} \bra{0_n}\!\tau_0\!\ket{0_n}^2 \!\! \sum_{\tau'\in \mathbb{P}^*_n}\! \bra{0_n}\!\tau'\!\ket{0_n}^4
\end{equation}
which becomes by the analogous argument as above
\begin{equation}
\mathbb{V}_\tau(m)\leq \frac{2^{3n}(2^n+1)^3}{2^{3n}(2^{2n}-1)} = O(2^n),
\end{equation}
as intended.

\section{Marginal channel reconstructions and cross-talk tomography}
\label{sec:details_cross-talk}

We here show how to employ the local Clifford UIRS protocol to get tomographic information of channel marginals.
To this end, recall that with the local Clifford UIRS protocol we can efficiently estimate the quantity $3^{-|w|}\tr(\Lambda P_w UP_w\Lambda)$ for any unitary channel $U$, where $\Lambda$ is a quantum channel and $P_w$ is the projector onto the irreducible representation of $\mathbb{C}_1^{\times n}$ labeled by the bit string $w$, provided $|w|$ bounded.
We can introduce channel marginals $\Lambda_k$ of $\Lambda$ by inserting a maximally mixed state into all but the first $k$ (out of $n$) inputs and  tracing out all but the first $k$ output qubits. Note that we can choose the order of the qubits arbitrarily, therefore restricting to the first $k$ qubits does not cost any generality.
For any $w\in \{0,1\}^k \times \{0\}^{n-k}$ we have
 \begin{equation}
 3^{-|w|}\tr(\Lambda P_w AP_w) = 3^{-|w|}\tr(\Lambda_k P_w AP_w),
 \end{equation}
 which only depends on the marginal $\Lambda_k$.
Now consider the $k$-qubit super-operator
 \begin{equation}
 \label{eq:pinched_marginal}
 S_k = \sum_{w\in \{0,1\}^k} P_w \Lambda_k P_w,
 \end{equation}
 which we will refer to as the \emph{pinched marginal} associated with the marginal $\Lambda_k$.
 Note that this super-operator is not necessarily a quantum channel (although it is trace preserving).
 One can see that the pinched marginal is composed of blocks $\Lambda_w = P_w\Lambda P_w$
 which we refer to as the \emph{unital marginals} in the main text.

 We can reconstruct the pinched marginal $S_k$ using the local Clifford UIRS protocol. To see this, consider the group of $k$-qubit Clifford operators $\mathbb{C}_k$. Reference~\cite[theorem 39]{AverageGateFidelities}
 implies that
 \begin{equation}
 \frac{1}{\lvert \mathbb{C}_k \rvert} \sum_{C\in \mathbb{C}_k} \big((2^{2k}-1)\tr(S_k \omega(C)\ct) - (2^{2k}-2)\big)\omega(C)  = S_k
 \end{equation}
 using that the Clifford group is a $2$-design (on $k$ qubits). Using the definition of $S_k$ we see that
 \begin{equation}
 \tr(S_k \omega(C)) = \sum_{w\in\{0,1\}^k} 3^{|w|}\big(3^{-|w|} \tr(\Lambda_k P_w CP_w)\big).
 \end{equation}
 From \cref{thm:local_cliff_var_sm} we know that we can estimate the quantities $3^{-|w|} \tr(\Lambda_k P_w CP_w)$ to accuracy $\epsilon$ using $S = O(k 2^{(2 + 2\log_2(3))k} / \epsilon^2)$ runs of the local Clifford UIRS protocol.
 Hence, we can reconstruct $S_k$ to $\epsilon$ error in diamond norm using $S = O(k 2^{(2 + 4\log_2(3))k} / \epsilon^2)$ runs.
In particular, we can also construct every `block' $\Lambda_w$ of $S_k$ with additive error in diamond norm from the same number of samples.
 Moreover, since the procedure we have described above is independent of which set
 of $k$ qubits is considered, it follows immediately that one can reconstruct all $\binom{n}{k}$ pinched marginals associated to each set of $k$  qubits to a global $\epsilon$ error in diamond norm using $S = O(n h(k/n) k 2^{(2 + 4\log_2(3))k} / \epsilon^2)$ samples (where $h(k/n)$ is the binary entropy).
Using $\log\binom{n}{k} \leq  k \log(\e n/k)$, we can relax the statement to guarantee $\epsilon$-accurate recovery of all unital marginals $\Lambda_w$ with $|w|=k$ in diamond norm from $O(k^2 2^{9k} /\epsilon^2)$.

\section{Details on SPAM-robust channel reconstruction}\label{appsec:reconstruction}

Using multi-qubit Clifford UIRS we can extract the relative average-gate fidelities that enter the tomographic reconstruction of both schemes from the output statistics of random gate-set sequences, without the need to perform different interleaved experiments. This gives rise to an efficient and robust channel reconstruction protocol.
In the multi-qubit Clifford UIRS protocol the decay rate are given as
\begin{equation}
p(A) = \frac{\tr(A\ct\Lambda)}{2^{2n}-1} .
\end{equation}
Hence, if we assume that $A = P_{\rm ad}U P_{\rm ad}$ for a unitary channel $U$ we see that $p(U) = (2^n F(U, \Lambda)-1)/(2^n-1)$, where $F(U, \Lambda)$ is the average fidelity between $U$ and $\Lambda$.
By \cref{thm:var_unitary_sm} and \cref{thm:cliff_var_sm}, we can estimate using the UIRS protocol, an exponential number of average fidelities using only a polynomial number of samples (and equivalently channel queries).
Furthermore, for $U$ a Clifford unitary calculating the sequence correlation function, and, thus, the entire classical-post processing, is time and space efficient in the number of qubits.

Characterizing a quantum channels in terms of different relative average gate fidelities with unitaries can provide valuable diagnostic in itself.
This can be seen as a robust gate-set or channel variant of selective state tomography \cite{MorrisDakic:2019:Selective}.
Beyond this, building on the results of refs.~\cite{Sco08,KimmOhki,KimLiu16,AverageGateFidelities}, having access to relative average fidelities is an powerful primitive for the tomographic reconstruction of channels, which we will now consider in more detail.

The first task we consider is the reconstruction of unitary (or more generally bounded Kraus rank) quantum channels---low-rank randomized benchmarking tomography.
This task is vital to the characterization of calibration errors.
Reference~\cite{AverageGateFidelities} establishes that given a list of estimates $[\hat{F}(C,\Lambda)]_{C\in \mc A}$ of relative average gate fidelities with respect to a randomly chosen subset $\mc A$ of Clifford unitaries, a constraint least-squares fit can reconstruct $\Lambda$ provided that
$|\mc A| \geq c d^2 \log(d)$.
More precisely, the error of the channel estimate $\hat \Lambda$ in Hilbert-Schmidt norm of the Choi-states fulfills
\begin{equation}
  \norm{J(\Lambda) - J(\hat{\Lambda})}_{HS} \leq 2^{2n}\frac{\norm{{\hat{F}} - {F}}_{2}}{\sqrt{|S|} }
\end{equation}
where ${F}:=[F(C, \Lambda)]_{C\in \mc A}$ is the vector of average fidelities of length $|\mc A|$, $\hat{F}$ is an estimate of $F$ produced through the shadow sequence protocol and the norm is the $l_2$ vector norm.
Furthermore, the reconstruction is stable against $\Lambda$ deviating from the low-rank assumption (model-mismatch) and can be formulated in different $p$-norms on both sides, we refer to the supplemental material of ref.~\cite{AverageGateFidelities} for details.
Reference~\cite{AverageGateFidelities}, however, has not analyzed the overall sampling complexity of the resulting RB tomography scheme when combined with a robust way to acquire the relative average fidelities.

The UIRS protocol can provide the missing piece. After decay fitting, we can give estimates ${\hat{F}}$ for the vector of fidelities ${F}$ with error guarantee
\begin{equation}
  \norm{{\hat{F}} - F}_{\infty} = O\left(\sqrt{\frac{\log(2|\mc A|/\delta)}{S}}\right)
\end{equation}
with success probability $1-\delta$ using $S$ samples, i.e., the size of the gate-set shadow.
Using the standard relations between $l_\infty$ and $l_2$ vector norms this implies we can obtain an $\epsilon$-accurate reconstruction of $\Lambda$ provided
\begin{equation}
  S\geq C\, 2^{4n}\frac{\log(2|\mc A|/\delta)}{\epsilon^2}\,
\end{equation}
with a suitable constant $C$.
Dropping polynomial factors in $n$ and $1/\epsilon$, we find that the total number of
gate-set shadows scales as $O(2^{4n})$.
Note that the number of channel invocations is bounded by the maximal sequence length times the number of sequences.
This matches the scaling of the information theoretic lower bound derived in ref.~\cite{AverageGateFidelities} for the case that the average gate fidelities are measured independently.
Besides the favourable scaling, the UIRS protocol has the benefit compared to, e.g., the interleaved protocol of ref.~\cite{KimmOhki} that the same measurement data is used for estimating all the average fidelities.

Besides the compressive, low-rank quantum channel tomography, we can use average gate-fidelities from UIRS protocols for the tomography of a more general class of quantum channels.
If $\Lambda$ is a unital quantum channel then it is known~\cite[theorem 38]{AverageGateFidelities} (see also ref.~\cite{Sco08}) that it can be expressed as
\begin{equation}
  \Lambda = \frac{1}{S}\sum_{C\in S}\big(D F(C, \Lambda) - 2^{-n}D +1)\omega(C)
\end{equation}
with $D = 2^n(2^n+1)(2^{2n}-1)$, provided the set $S$ is a unitary $2$-design. We can provide a direct reconstruction for the unital channel $\Lambda$ by calculating
\begin{equation}
\hat{\Lambda} = \frac{1}{S}\sum_{C\in S}\big(D \hat{F}(C, \Lambda) - 2^{-n}D +1)\omega(C)
\end{equation}
where the estimates $\hat{F}(C, \Lambda)$ are again provided by the multi-qubit Clifford UIRS protocol.
The accuracy of this reconstruction depends on the metric used. If we
consider the Hilbert-Schmidt of the Choi state as before, we see that
\begin{align}
\norm{J(\Lambda) - J(\hat{\Lambda})}_{HS} &=\frac{D}{|S|} \sum_{C\in S}|F(C,\Lambda) - \hat{F}(C,\Lambda)|,%&\leq  D\norm{F-\hat{F}}_{\infty}
\end{align}
since $\norm{J(C)}_{HS} = 1$ for all unitaries $C$.
Hence, by the same argument as in the unitary case, the number of samples required scales as $O(2^{8n})$. The same argument holds for all norms for which $\norm{C}=1$, such as the diamond norm.

\end{widetext}

\bibliographystyle{apsrev4-1}
%\bibliography{BigReferencesRB}

\section*{Acknowledgments}
During the final stages of drafting this manuscript, we became aware of refs.~\cite{ProcessShadows2,ProcessShadows1}
which generalize state shadow estimation to quantum processes using the Choi-Jamiolkowski isomorphism, but without considering self-consistent sequences of gate-sets (see also remarks in ref.~\cite{RobustShadows} in this context).
The authors would like to thank Thomas Monz, Martin Kliesch and Richard Kueng for discussions.
J.~H. is supported by the Quantum Software Consortium Zwaartekracht grant.
The Berlin team has been funded by the BMBF (DAQC, MUNIQC-ATOMS), and
the Munich Quantum Valley (K-8). Funded also by the Deutsche Forschungsgemeinschaft (DFG, German Research
Foundation) (EI 519 20-1, CRC 183, Daedalus, as well as under Germany's Excellence Strategy - The Berlin Mathematics
Research Center MATH+, EXC-2046/1, project ID: 390685689).
It has also received funding from the EU's Horizon 2020 research and innovation program
 (PASQuanS, \newtext{PASQuanS2, Millenion}).
E.~O. is supported by the Royal Society, by the UK Hub in Quantum Computing and Simulation, part of the UK National Quantum Technologies Programme with funding from UKRI EPSRC (grant EP/T001062/1) and by the Bavarian state government with funds from the Hightech Agenda Bayern Plus as part of the Munich Quantum Valley.
A.~H.~W.~thanks the VILLUM
FONDEN for its support with a Villum Young Investigator
Grant (Grant No. 25452) and its support via the QMATH
Centre of Excellence (Grant No.~10059).

\section*{Author contributions statement}

\newtext{I.~R.~and J.~H.~conceived of the initial idea to the project, with substantial contributions from E.~O., A.~H., and J.~E. 
All authors contributed to devising the overall scheme, finding and exploring applications and conceptualizing the results. J.~H.~has taken the lead in proving performance bounds and writing an initial draft together with I.~R.\ 
M.~I.~proved the performance bounds for the local Clifford group, J.~K.~has performed the numerical analysis. 
All authors contributed substantially to the final manuscript.}

\section*{Competing interests statement}

\newtext{The authors declare no competing interests.}

\end{document}